%% file: main.tex
\documentclass[final,1p,
]{elsarticle}
\biboptions{sort&compress}
\usepackage{tikz-cd}
\usepackage{amsmath}
\usepackage{amsthm}
\usepackage{amssymb}
\usepackage{amsfonts}
\usepackage{graphicx}           
\usepackage{hyperref}
\hypersetup{
  colorlinks,
  citecolor=Blue,
  urlcolor=cyan
}
\usepackage{enumerate}
\usepackage{bm}
\usepackage{todonotes}
\usepackage{mathtools}
\usepackage{mathrsfs}
\usepackage{xcolor}
\usepackage{wrapfig}
\usepackage[binary-units = true]{siunitx}
\usepackage{makecell}
\usepackage[algo2e,boxed,vlined,linesnumbered,
titlenumbered]{algorithm2e}
\let\oldnl\nl
\newcommand{\nonl}{\renewcommand{\nl}{\let\nl\oldnl}}
\SetKw{KwFrom}{from}
\SetKwFor{For}{For}{do}{End for}
\SetKwFor{Forall}{For all}{do}{End for}
\SetKw{KwRet}{Return}
\SetKw{KwErr}{Error}
\SetKw{KwAnd}{and}
\SetKw{KwOr}{or}
\SetKw{KwBreak}{break}
\SetKw{KwContinue}{continue}
\SetKw{KwTrue}{true}
\SetKw{KwFalse}{false}
\SetKw{KwNext}{next}
\SetKwFor{While}{While}{do}{End do}
\SetKwIF{If}{ElseIf}{Else}{If}{then}{Else if}{Else}{End if}
\SetCommentSty{mycommfont}

\makeatletter
\newtheorem*{rep@theorem}{\rep@title}
\newcommand{\newreptheorem}[2]{%
  \newenvironment{rep#1}[1]{%
    \def\rep@title{#2 \ref{##1}}%
    \begin{rep@theorem}}%
    {\end{rep@theorem}}}
\makeatother

\makeatletter
\def\thmhead@plain#1#2#3{%
  \thmname{#1}\thmnumber{\@ifnotempty{#1}{ }\@upn{#2}}%
  \thmnote{ {\the\thm@notefont#3}}}
\let\thmhead\thmhead@plain
\makeatother
\newtheorem{theorem}{Theorem}
\newtheorem{lemma}[theorem]{Lemma}
\newtheorem{corollary}[theorem]{Corollary}

\newtheorem{proposition}[theorem]{Proposition}
\newreptheorem{theorem}{Theorem}

\newtheorem{notation}[theorem]{Notation}
\newtheorem{example}[theorem]{Example}
\newtheorem{remark}[theorem]{Remark}

\newcounter{algo}

\newcommand\f{\bm{f}}
\newcommand\g{\bm{g}}
\newcommand\N{\mathbb{N}}
\newcommand\Q{\mathbb{Q}}

\newcommand\R{\mathbb{R}}
\newcommand\C{\mathbb{C}}
\newcommand\K{\mathbb{K}}

\newcommand\V{\mathbf{V}}

\newcommand\G{\mathbb{G}}

\newcommand\cG{\mathcal{G}}
\newcommand\FGb{\textsc{FGb}\xspace}
\newcommand\MSolve{\textsc{MSolve}\xspace}
\newcommand\Maple{\textsc{Maple}\xspace}

\newcommand\gb{Gr\"obner basis\xspace}
\newcommand\gbs{Gr\"obner bases\xspace}
\newcommand\geo{geometric resolution\xspace}
\newcommand\cO{\mathcal{O}}
\newcommand\bv{\mathbf{v}}
\newcommand\bw{\mathbf{w}}
\newcommand\bz{\mathbf{z}}
\newcommand\bx{\mathbf{x}}
\newcommand\bc{\mathbf{c}}
\newcommand\br{\mathbf{r}}
\newcommand\cOE{\mathcal{O}_\mathcal{E}}
\newcommand\cOA{\mathcal{O}_{\GL}}
\newcommand\alg{acv0\xspace}
\newcommand\Mult{\mathsf{M}}

\DeclareMathOperator\dist{dist}
\DeclareMathOperator\jac{jac}
\DeclareMathOperator\GL{GL}

\DeclareMathOperator\gr{graph}
\DeclareMathOperator\numer{numer}

\DeclareMathOperator\Saturate{Saturate}
\DeclareMathOperator\Eliminate{Eliminate}
\DeclareMathOperator\Intersect{Intersect}

\DeclarePairedDelimiter{\ideal}{\langle}{\rangle}
\DeclarePairedDelimiter{\norm}{\|}{\|}

\DeclareMathOperator\ACV{ACV}
\DeclareMathOperator\rank{rank}
\DeclareMathOperator\res{Res}

\renewcommand{\d}{\ensuremath{\operatorname{d}\!}}
\newcommand{\ord}{\ensuremath{\operatorname{ord}\!}}

\begin{document}

\begin{frontmatter}

  \title{On the computation of asymptotic critical values of \\ polynomial maps and applications}

  \author{J\'er\'emy Berthomieu}
  \ead{jeremy.berthomieu@lip6.fr}
  \author{Andrew Ferguson}
  \ead{andrew.ferguson@lip6.fr}
  \author{Mohab Safey El Din}
  \ead{mohab.safey@lip6.fr}
  \address{Sorbonne Universit\'e, \textsc{CNRS}, \textsc{LIP6}, F-75005, Paris,
    France}
  \begin{abstract}
    Let $\bm{f} = \left(  f_1, \dots, f_p\right) $ be a polynomial tuple in
    $\mathbb{Q}[z_1, \dots, z_n]$ and let $d = \max_{1 \leq i \leq p} \deg f_i$. 
    We consider the problem of computing the set of asymptotic critical values of
    the polynomial mapping, with the assumption that this mapping is dominant, 
    $\bm{f}: z \in \mathbb{K}^n \to (f_1(z), \dots, f_p(z)) \in \mathbb{K}^p$
    where $\mathbb{K}$ is either $\mathbb{R}$ or $\mathbb{C}$. This is the set
    of values $c$ in the target space of $\bm{f}$  such that there exists a sequence of
    points $(\bx_i)_{i\in \mathbb{N}}$ for which $\bm{f}(\bx_i)$ tends to $c$ and $\|\bx_i\|
    \kappa(\ensuremath{\operatorname{d}\!} \bm{f}(\bx_i))$  tends to $0$ when $i$ tends to
    infinity where $\ensuremath{\operatorname{d}\!} \bm{f}$ is the differential of $\bm{f}$
    and $\kappa$ is a function measuring the distance of a linear operator to the set of
    singular linear operators from $\mathbb{K}^n$ to $\mathbb{K}^p$.

    Computing the union of the classical and asymptotic critical values allows one to put
    into practice generalisations of Ehresmann's fibration theorem. This leads to natural
    and efficient applications in polynomial optimisation and 
    computational real algebraic geometry.

    Going back to previous works by Kurdyka, Orro and Simon, we design new algorithms to
    compute asymptotic critical values. Through randomisation, we introduce new geometric
    characterisations of asymptotic critical values. This allows us to dramatically reduce the
    complexity of computing such values to a cost that is essentially $O(d^{2n(p+1)})$ arithmetic
    operations in $\mathbb{Q}$. We also obtain tighter degree bounds  on a hypersurface containing
    the asymptotic critical values, showing that the degree is at most $p^{n-p+1}(d-1)^{n-p}(d+1)^{p}$. 
    
    Next, we show how to apply these algorithms to unconstrained polynomial optimisation problems and
    the problem of computing sample points per connected component of a semi-algebraic set defined by a
    single inequality/inequation.

    We report on the practical capabilities of our implementation of this algorithm. It shows 
    how the practical efficiency surpasses the current state-of-the-art algorithms for computing
    asymptotic critical values by tackling examples that were previously out of reach.
  \end{abstract}

  \begin{keyword}
    Asymptotic critical values, Polynomial optimisation, Gr\"obner bases
  \end{keyword}
\end{frontmatter}

\input{newintroduction.tex}
\input{newpreliminaries.tex}
\input{newfaster_algo.tex}
\input{proof_degree.tex}
\input{proof_comp.tex}
\input{apps.tex}
\input{experiments.tex}

\paragraph*{Acknowledgements.} The authors are supported by the ANR grants
ANR-18-CE33-0011 \textsc{Sesame}, ANR-19-CE40-0018 \textsc{De Rerum Natura} and
ANR-19-CE48-0015 \textsc{ECARP}, the PGMO grant \textsc{CAMiSAdo} and the
European Union's Horizon 2020 research and innovation programme under the Marie
Sklodowska-Curie grant agreement N. 813211 (POEMA).

\bibliographystyle{elsarticle-harv}
\bibliography{lemma}
\end{document}

%% file: newintroduction.tex
\section{Introduction}
\label{sec:intro}

\paragraph*{Basic definitions and problem statement}
Let $\bz$ denote $z_1, \dots, z_n$. Let $\f = \left( f_1, \ldots, f_{p} \right)$ be a polynomial
tuple in $\K[\bz]$ where $\K$ is either $\R$ or $\C$ and let
$d = \max_{1 \leq i \leq p} \deg f_i$.

By a slight abuse of notation, we will also denote by $\f$,
the polynomial mapping \[
    \f: \bx = \left(x_1, \dots, x_n  \right) \mapsto \left( f_1(\bx), \ldots,
    f_p\left( \bx \right)  \right)  
  \]
and we assume this mapping is dominant.
Furthermore, denote by $\d \f$ the differential of the mapping $\f$ and, for a given point $\bx \in \C^n$,
$\d \f(\bx)$ the differential of $\f$ at $\bx$, a linear map from $\K^n$ to $\K^p$.  

Let $\jac(\f)$ be the Jacobian matrix associated to $\f$
 \[
     \jac(\f) = \left[ \begin{array}{ccc}
             \frac{\partial f_1}{\partial z_1} & \cdots & \frac{\partial
             f_1}{\partial z_n} \\
                 \vdots & & \vdots \\
                 \frac{\partial f_1}{\partial z_1} & \cdots & \frac{\partial
             f_1}{\partial z_n} \\
     \end{array} \right] 
 .\] For $1 \le  j \le  p$, we denote by $\jac(\f^{[j]})$ the submatrix obtained 
 by removing the $j$th column of $\jac(\f)$. Let $w_j(z)$ be the
 restriction of $\d f_j$, the differential of the polynomial $f_j$, to
 the kernel of  $\jac(\f^{[j]})(z)$. When $j=p=1$, this matrix is empty.
 By convention, we say its kernel is $\K^n$.

\begin{example}\label{ex:ej}
  Let $f = z_1^4 + (z_1z_2-1)^2$. Since $p=1$, $\ker\jac(f^{[1]})=\K^2$.
  Thus, $w_1(z)$ is simply the gradient of $f$:
  \[
    w_1(z) = \left(4z_1^3 + 2z_2(z_1z_2-1), 2z_1(z_1z_2-1)\right).
  \]

  Let $\g = \left(z_1z_2,z_1z_3\right)$. Then, the Jacobian matrix associated to $\g$ is:
  \[
     \jac(\g) = \left[ \begin{array}{ccc}
                         z_2 & z_1 & 0 \\
                         z_3 & 0 & z_1 \\
     \end{array} \right] 
   .\]
 Therefore, $\jac(\g^{[1]}) = \left[ z_3, 0, z_1 \right]$ and so we restrict the linear mapping
 \[
   \d g_1(z): \gamma = (\gamma_1,\gamma_2,\gamma_3) \mapsto \gamma_1z_2 + \gamma_2z_1
 \]

 to the space $\{\gamma \in \K^3 \; | \; \gamma_1z_3 + \gamma_3z_1 = 0\}$, the result being $w_1(z)$.
 The construction of $w_2(z)$ follows similarly.
\end{example}

 Following~\cite[Definition 2.2]{KOS}, for $z
 \in\K^n$, we consider the so-called Kuo distance
  \[
    \kappa(\d \f(z)) = \min_{1 \leq j \leq p} ||w_j(z)||.
  \]
  Then, the set of \textit{asymptotic critical values} of the 
  mapping $\f$ is defined to be the set: 
  \[
    K_{\infty} (\f) = \left\{ c\in \C^p  \mid \exists
      (\bx_t)_{t\in \N} \subset \C^n \; \text{s.t.}\; \norm{\bx_t}\to\infty, \f(\bx_t)\to c
      \text{ and } \|\bx_t\|\kappa(\d\f(\bx_t))\to 0 \right\}.
  \] 

By \cite[Theorem 4.1]{KOS}, $K_{\infty}(\f)$ is an algebraic set of
$\C^p$.  
  We also define the function $\nu$ introduced by Rabier in
  \cite{rabier}, given in \cite[Definition 2.1]{KOS}.
  Denote by $L(\K^n,\K^p)$ the space of 
  linear mappings from $\K^n$ to $\K^p$ and by $\Sigma$ the singular set of $L(\K^n,\K^p)$.
  The distance function to $\Sigma$ is given by \[
    \dist(A,\Sigma) = \inf_{B \in \Sigma} \norm{A-B}.
  \]
  For $A \in L(\K^n,\K^p), A^*$ denotes the adjoint operator of $A$. For example, in the setting of
  complex numbers, the adjoint operator is the complex conjugate transpose. Then,
  \[\nu(A) = \inf_{\norm{\phi} = 1} \norm{A^* \phi}.\]
  We give the following two properties of $\nu$ given in \cite{KOS} that will be used in the proof
  of correctness of our algorithms.
  By \cite[Proposition 2.2]{KOS}, \[ \nu(A) = \dist(A, \Sigma). \] Furthermore, by
  \cite[Corollary 2.1]{KOS}, the functions $\nu$ and $\kappa$ are equivalent in the following sense:
  For $A \in L(\K^n,\K^p)$, \[\nu(A) \leq \kappa(A) \leq \sqrt{p} \nu(A).\]
  Therefore, as in the definition given by Rabier in \cite{rabier}, we can also define the
  asymptotic critical values in terms of the function $\nu$.
 \[
    K_{\infty} (\f) = \left\{ c\in \C^p  \mid \exists
      (\bx_t)_{t\in \N} \subset \C^n \; \text{s.t.}\; \norm{\bx_t}\to\infty, \; \f(\bx_t)\to c\;
      \text{ and }\;  \|\bx_t\|\nu(\d\f(\bx_t))\to 0 \right\}.
  \]
  However, we will predominantly use the prior definition of the set of asymptotic critical values.
\begin{example}\label{ex:pop1}
  Consider the polynomial $f = z_1^4 + (z_1z_2-1)^2$. This polynomial has a critical
  point at $(0,0)$ that corresponds to the critical value
  $1$. However, it is easy to see that $f$ takes values less than $1$. Take
  the path parameterised by $t$: $z_1(t) = \frac{1}{t}$, $z_2(t) = t$. Then,
  as $t \rightarrow \infty$ we see $f(t) \rightarrow 0$. From Example~\ref{ex:ej}, we
  know that:
  \[
    w_1(z) = \left(4z_1^3 + 2z_2(z_1z_2-1), 2z_1(z_1z_2-1)\right).
  \]
  Then, along the path parameterised by $t$:
  \[
    (z_1^2(t) + z_2^2(t))w_1(z(t)) = \left(\frac{4}{t^5} + \frac{4}{t}, 0\right) \to (0,0) \text{ as } t \to \infty.
  \]
  Therefore, $0$ is an asymptotic critical value of $f$.
\end{example}

The goal of this paper is to provide efficient algorithms which, on input
$\f$, compute finitely many polynomials whose simultaneous vanishing set
contains $K_{\infty}(\f)$.

\paragraph*{Motivations and prior works}
Denote by $K_0(\f)$, the set of critical values of $\f$ 
\[
K_0(\f) = \left\{ c\in \C^p  \mid \exists \bx \in \C^n \text{ s.t. } \f(\bx)=c \text{ and
  }\rank(\d\f(\bx)) < p  \right\}.
\]
The set of generalised critical values is defined as the union 
$K_0(\f) \cup K_{\infty}(\f)$; it will be denoted by $K(\f)$. 
One strong property of $K(\f)$ is the following. The
mapping $\f$ restricted to $\K^n \setminus \f^{-1}(K(\f))$ is a locally trivial fibration. Thus,
for all connected open sets $U \subset \K^p \setminus K(\f)$, for all $y \in U$ there exists a
diffeomorphism $\varphi$ such that, with $\pi$ as the canonical projection map, the following diagram
commutes~\cite[Theorem 3.1]{KOS}.
\begin{center}
  \begin{tikzcd}
    \f^{-1}(y) \times U \arrow{rd}{\pi} \arrow{r}{\varphi} 
    & \f^{-1}(U) \arrow{d}{\f} \\
    & U
  \end{tikzcd}
\end{center}
In other words, this definition of
generalised critical values allows one to generalise Ehresmann's fibration
theorem to non-proper settings. This goes back to the problem of defining sets that
contain the so-called bifurcation set of $\f$ as introduced
in~\cite{rabier}. From that perspective, one crucial feature of
$K_{\infty}(\f)$, is that there is a generalised Sard's theorem for
this set, i.e.\ 
the codimension of $K_{\infty}(\f)$ is greater than or equal to
one~\cite[Theorem 3.1]{KOS}. 

These properties make the effective use of generalised critical values quite
appealing in computational real algebraic geometry. In particular, 
in~\cite{greuetsafey,msed2007testing},
algorithms for 
\begin{itemize}
\item computing an exact representation of the minimum of a given
  polynomial (i.e.\ the minimal polynomial of this minimum and an
  isolating interval), 
\item computing sample points for each connected component of a
  semi-algebraic set defined by a single inequality, 

\end{itemize}
have been designed, relying on the computation of generalised critical values
when $p=1$. In this context, $\f$ is assumed to lie in  $\Q[\bz]$ and, in the works we refer to below, the chosen complexity model is
the arithmetic one, i.e.\ one counts arithmetic operations in the base field
$\Q$ without taking into account the growth of the bit sizes of the
coefficients. We use the classical big-O notation $O \left( \phi(\bx)
\right)$ (where $\phi$ is a real valued function that is strictly
positive for all large enough values of $\bx$) to
denote the class of non-negative functions which up to a multiplicative constant, are
bounded from above by $\phi(\bx)$ at infinity. 

As far as we know, the first work that leads directly to an algorithm for computing $K_{\infty}(\f)$ is
given in~\cite{KOS}. It is based on a geometric characterisation of
$K_{\infty}(\f)$ that allows one to apply algebraic elimination algorithms
(e.g.\ algorithms for computing elimination ideals in polynomial rings) to get
an algebraic description for $K_{\infty}(\f)$. The
tool of choice for performing algebraic elimination in \cite{KOS,K2} is
Gr\"obner bases. In Section~\ref{sec:theorem}, we recall the geometry involved in this algorithm.
At this stage, let us say that it builds equations defining
locally closed sets in $\C^{n + p(n+2)}$. 

It then considers the intersection of these sets
with some linear subspaces chosen in such a way that the union of the
projections of those intersections on the target
space of $\f$ contains $K_{\infty}(\f)$.  

Several attempts to improve this algorithmic pattern have been made. When
$p=1$, we mention~\cite{msed2007testing} 
which makes the connection between generalised critical values and
properties of polar varieties. This yields a probabilistic algorithm whose
runtime is in the order of $d^{O(n)}$ arithmetic operations in $\Q$.
However, as noticed in~\cite{polarcurve}, the argumentation
in~\cite{msed2007testing} is incomplete. Still, this connection is
exploited through the use of finite dimensional spaces of rational arcs to
compute generalised critical values in~\cite{jelonek2014reaching,polarcurve}.  
Nevertheless, the complexity estimates in~\cite{msed2007testing} do not apply
to~\cite{jelonek2014reaching,polarcurve}. 

It should also be noted that the result in~\cite{KOS} was generalised
in~\cite{K2}. This allowed the design of an algorithm to compute the
generalised critical values of a polynomial mapping restricted to an
algebraic set, a setting not covered in this paper. 

\paragraph*{Main results}
We build upon the geometric characterisations of $K_{\infty}(\f)$~\cite{KOS} to
obtain new ones that allow us to design more efficient algorithms.

We introduce an element of randomisation to avoid some combinatorial
steps in the algorithm designed in~\cite{K2}. Next, we introduce another element of randomisation
that reduces the computation of $K_\infty(\f)$ to intersecting the Zariski
closure of some locally closed subset of $\C^{n+p+1}$ with a linear affine
subspace of codimension $2$ such that the projection onto the
target space of $\f$ of this intersection contains $ K_{\infty} (\f)$.

The sets involved in the intersection process described above to
geometrically characterise $K_\infty(f)$ can be related to some incidence varieties. 

We use this relation to design another algorithm that takes advantage of the
determinantal setting. This allows us to obtain a faster algorithm in practice, though the theoretical
complexity of both algorithms is essentially the same. We can now state our first main result.

\begin{theorem}\label{prop:degree}
  Let $\f = (f_1, \ldots, f_p) \in \K[\bz]^p$ be a
  dominant polynomial mapping
  and let $d = \max_{1 \leq i \leq p} \deg f_i$. 
  Then, the asymptotic critical values of $\f$ are contained in a
  hypersurface of degree at most $p^{n-p+1}(d-1)^{n-p}(d+1)^{p}$.
\end{theorem}

We study the complexity of resulting algorithms by substituting Gr\"obner bases computations with
the use of the geometric resolution algorithm designed in~\cite{GML}. We first introduce some notation.

Given a polynomial sequence $\g$ in  $\mathbb{F}[z_1, \ldots, z_m]$, where
$\mathbb{F}$ is a field and $\overline{\mathbb{F}}$ is an algebraic closure of $\mathbb{F}$
, we denote by $\V(\g) \subset \overline{\mathbb{F}}^m$ the algebraic set defined by the
simultaneous vanishing of the entries of $\g$. We also recall the ``soft-Oh'' notation:
$f(n)\in\tilde{O}(g(n))$ means that $f(n)\in g(n)\log^{O(1)}(3+g(n))$,
see also~\cite[Chapter~25, Section~7]{GathenG2013}. Additionally, denote by $\bc$ the
indeterminates $c_1, \dots, c_p$.

\begin{theorem}\label{prop:comp}
  Let $\f = (f_1, \ldots, f_p) \in \K[\bz]^p$ be a
  dominant polynomial mapping 
  and let $d = \max_{1 \leq i \leq p} \deg f_i$.
  
  There exists an algorithm which, on input $\f$, computes a
  non-zero polynomial $g$
  in $\K[\bc]$ such that $ K_{\infty}(\f) \subset \V(g)$ using at most 
  \[
    O^{\sim}\left( p (p(d-1))^{2(p+1)(n-p)} (d+1)^{2p(p+1)} \right) 
  \]
  arithmetic operations in $\K$. 
\end{theorem}

Following the algorithmic scheme of \cite{msed2007testing}, we show how to
apply these new algorithms for computing asymptotic critical values, to the
problem of computing sample points per connected component of a
semi-algebraic set defined by a single inequality. Furthermore, we show how to
use the generalised critical values to tackle the problem of computing an exact
representation of the global infimum of a given polynomial. Note that this
algorithm can decide if this infimum exists in $\R$ or if the given polynomial
is unbounded from below.

We implemented all the aforementioned algorithms for computing asymptotic
critical values in the \textsc{Maple}
computer algebra system where we substitute the geometric resolution algorithm
for Gr\"obner bases. For Gr\"obner bases computations, we rely on a combination of
\MSolve~\cite{MSolve} and the \textsc{FGb} library~\cite{FGb}. 
We used an extensive set of benchmark examples to illustrate the computational capabilities of these
algorithms. 

It appears that our new algorithms outperform the state-of-the-art and can
tackle examples which were previously out of reach.

\paragraph*{Structure of the paper}
In Section~\ref{sec:theorem}, we revisit the geometric characterisation of
asymptotic critical values given in~\cite{KOS} and describe our first instance
of randomisation. One can derive from this result, an algorithm that is similar
to the deterministic one derived from~\cite{KOS}, save for a combinatorial factor.
Next, in Section~\ref{sec:impro}, we go further and design two new algorithms, more efficient than
the state-of-the-art, on which Theorem~\ref{prop:degree} and Theorem~\ref{prop:comp} rely. These
theorems are then proved in Section~\ref{sec:proof}. Section~\ref{sec:apps} returns to the applications
given in~\cite{greuetsafey,msed2007testing} in the context of the newly designed algorithms.
We conclude with Section~\ref{sec:expo}, which gives the results of our experiments with our algorithms
for several families of polynomials, as well as for polynomials found in practice.

%% file: newpreliminaries.tex
\section{Preliminaries}
\label{sec:theorem}

Let $\f \in \K[\bz]^p$ be a dominant polynomial mapping. Firstly, let us recall that
by~\cite[Theorem~3.1]{KOS}, $K_\infty(\f)$ has codimension at least $1$ in $\C^p$.
Then, using the geometric description of $K_\infty(\f)$ provided in
the proof of~\cite[Theorem~4.1]{KOS}, one can compute it.
For the sake of completeness, we follow the proof given
in~\cite[Theorem 4.1]{KOS} with slight modifications to~\cite[Lemma
4.1]{KOS} that are useful for the algorithm we describe in the next
section.

\begin{notation}\label{not:ksat}
Firstly, we introduce a few objects. We shall make use of the following change of coordinates to handle the
asymptotic behaviour, sending $z_s = 0$ to $\infty$:
\[
  \tau_s(z) = \left( \frac{z_1}{z_s}, \dots, \frac{z_{s-1}}{z_s}, \frac{1}{z_s}, \frac{z_{s+1}}{z_s}, \dots, \frac{z_n}{z_s} \right).
\]

For each choice of $s=1, \ldots,n$, $j=1, \ldots, p$ and point $\bx \in \K^n$, let
$W_s^j(\bx)$ be the graph of $x_s w_j(\bx)$, a point in the Grassmannian of linear subspaces of
$\C^n \times \C$ that are of dimension $n-p+1$, denoted by $\G_{n-p+1}(\C^n \times \C)$.
Then, we consider the rational mappings
\begin{align*}
  M_s^j(\f) : \K^n \setminus \{z_s = 0\} & \rightarrow \C^p \times \G_{n-p+1}(\C^n \times \C), \\
  z & \mapsto (\f(\tau_s(z)), W_s^j(\tau_s(z))).
\end{align*}
The point $W_s^j(z)$ is well-defined for $z$ such that the kernel of $\jac(\f^{[j]})$ has 
dimension $n-p+1$. Since we assume $\f$ is dominant,
$M_s^j(\f)$ is well-defined outside of a proper Zariski closed subset of $\K^n$.

Let $\Lambda = \G_{n-p+1}(\C^n \times0)$. This is the set of $(n-p+1)$-dimensional
graphs of linear maps from $\C^n$ to $\C$ that are identically the zero map.
We then intersect the Zariski closure of the graph of $M_s^j(\f)$ in the following way:

\[
  L_s^j(\f) = \overline{\gr M_s^j(\f)} \cap (\{z \in \C^n | z_s = 0\} \times \C^p \times \Lambda).
\]

Define $\pi: \C^n \times \C^p \times \G_{n-k+1}(\C^n \times \C) \to \C^p$ to be the projection map and take $K_s^j(\f) = \pi(L_s^j(\f))$.
\end{notation}
By~\cite[Lemma 4.1]{KOS}, \[ K_\infty(\f) = \bigcup_{(s,j) = (1,1)}^{(n,p)} K_s^j(\f).\]  

\begin{example}\label{ex:msj}
  Let $f = z_1^4 + (z_1z_2-1)^2$. As in Example~\ref{ex:ej}, we have
  \[
    w_1(z) = \left(4z_1^3 + 2z_2(z_1z_2-1), 2z_1(z_1z_2-1)\right).
  \]
  Then,
  \[
    \tau_1(z) = \left(\frac{1}{z_1}, \frac{z_2}{z_1}\right), \; \tau_2(z) = \left(\frac{z_1}{z_2}, \frac{1}{z_2}\right).
  \]
  We now illustrate the construction of $K_1^1(f)$, taking care to exclude the variety $\V(z_1)$
  where the mapping is not defined.

  Firstly, we derive equations for $M_1^1(f)$ by evaluating $f$ and the linear map $w_1$ at $\tau_1(z)$,
  \[
    M_1^1(f) = \left(\frac{1}{z_1^4} + \left(\frac{z_2}{z_1^2} - 1\right)^2, \frac{4}{z_1^4} + \frac{2z_2^2}{z_1^3} - \frac{2z_2}{z_1^2}, \frac{2z_2}{z_1^4} - \frac{2}{z_1^2}\right).
  \]
  Then, we consider the graph of these functions and so we introduce the variables $c, u_1$ and
  $u_2$. To describe this set algebraically, we consider the numerators of these rational functions with
  their respective value variables and remove the variety $\V(z_1)$ where $M_1^1(f)$ is not defined.
  This leads to the following description of $\overline{\gr M_1^1(f)}$
  \[
    \overline{\V(c z_1^4 - z_1^4 + 2z_1^2z_2 - z_2^2 - 1, u_1z_1^4 + 2z_1^2z_2 - 2z_2^2 - 4, u_2z_1^4 + 2z_1^2-2z_2) \setminus \V(z_1)}.
  \]
  One can compute a finite list of polynomials whose simultaneous vanishing set is the above variety by
  a range of methods including Gr\"obner bases. Such methods are discussed in Subsection~\ref{ss:first}.
  Then, by setting $z_1 = u_1 = u_2 = 0$, one retrieves an algebraic description of the set
  $L_1^1(f)$. In this case, we do not find any asymptotic critical values for $s=1$ as the set
  $L_1^1(f)$, and therefore, the set $K_1^1(f)$, is empty.

  Nonetheless, we consider $K_2^1(f)$. As before, we introduce the variables $c, u_1$ and
  $u_2$ and consider the Zariski closure $\overline{\gr M_2^1(f)}$ expressed by:
  \[
    \overline{\V(z_1^4 + z_2^4 - 2z_1z_2^2 + z_1^2 - c z_2^4, 4z_1^3 - 2z_2^2 + 2z_1 - u_1z_2^3, 2z_1^2 - 2z_1z_2^2 - u_2z_2^3) \setminus \V(z_2)}.
  \]
  
  We perform a Gr\"obner basis computation and intersect with the variety $\V(z_2, u_1, u_2)$ to get
  the algebraic description of $L_2^1(f)$:
  \[
    L_2^1(f) = \V(z_1,z_2,u_1,u_2,c).
  \]
  We arrive at the Zariski closure of $K_2^1(f)$ by excluding all polynomials that involve any
  variables other than $c$. In this case, we see that $\overline{K_2^1(f)} = \V(c)$ and therefore $0$
  is the only possible asymptotic critical value of $f$. In combination with Example~\ref{ex:pop1}, we
  conclude that $K_\infty(f) = \{ 0\}$.

  Let $\g = \left(z_1z_2,z_1z_3\right)$. We investigate the case $s=1, j=1$, the other cases follow
  similarly. From Example~\ref{ex:ej}, we have that   \[
     \jac(\g) = \left[ \begin{array}{ccc}
                         z_2 & z_1 & 0 \\
                         z_3 & 0 & z_1 \\
     \end{array} \right].\]
 Therefore, $\jac(\g^{[1]}) = \left[ z_3, 0, z_1 \right]$. We consider the points $\bx \in \K^3$ such that
 the matrix $\jac(\g^{[1]})(\bx)$ has maximal rank. Thus, excluding evaluations at points in the variety
 $\V(z_1z_3)$, the kernel of the linear mapping has dimension $2$.

 It is easy to see that for such a fixed $\bx$, the kernel of
 $\jac(\g^{[1]})(\bx)$ is spanned by 
 \[B=\left(
     \begin{pmatrix}
       -x_1/x_3 \\
       0 \\
       1 \\
     \end{pmatrix},
     \begin{pmatrix}
       0 \\
       1 \\
       0 \\
     \end{pmatrix}\right).
 \]
           
 We can then describe $w_1(z)$ as the product of  $\d g_j$ with this basis evaluated in $z$:
 \[
   w_1(z) = \d g_1(z) B(z) = \left( \frac{-z_2z_1}{z_3}, z_1 \right).  
 \]

 Thus, introducing the independent variables $c_1, c_2, u_1$ and $u_2$ and applying the transformation $\tau_1$, we
 consider the the Zariski closure:
 \[
   \overline{\gr M_1^1(\g)} = \overline{\V(c_1z_1^2 - z_2, c_2z_1^2 - z_3, u_1z_1^2z_3 - z_2, u_2z_1^2 - 1) \setminus \V(z_1z_3)}.
 \]
 After performing a Gr\"obner basis computation to compute a finite list of polynomials describing this
 set, we intersect with the variety $\V(z_1, u_1, u_2)$. We arrive at an algebraic description of
 $L_1^1(\g)$:
 \[
   L_1^1(\g) = \V(1) = \emptyset.
 \]
 Thus, we find that $\overline{K_1^1(\g)}$ is empty. Through similar analysis of the cases $s = 2,3$ for $j=1,2$,
 one finds that $\overline{K_s^j(\g)} = \V(c_1,c_2)$. Therefore, $(0,0)$ is the only possible
 asymptotic critical value of $\g$.
\end{example}

We extend this algebraic description of the asymptotic critical values with the following lemma.
This lemma derives from~\cite[Lemma 4.1]{KOS} with one major difference. 
We introduce a non-empty Zariski open subset $\cOA$ of $\GL_n(\K)$, the group of $n \times n$
invertible matrices with entries in $\K$.
Essentially, we choose a random linear change of variables $A$, so that $A$ almost surely lies in
$\cOA$, and consider the
polynomial $\f^A$ given by $\f^A(z) = \f(A z)$. By~\cite[Lemma~2.4]{msed2006algo}, $K_\infty(\f^A) =
K_\infty(\f)$, and so we can compute the asymptotic critical values of $\f$ by computing those of
$\f^A$. We exploit this result in our algorithm by showing that choosing $A \in \cOA$ implies that for
$1\leq s\leq n$, whenever $z_s$ goes to $\infty$ in a path towards an
asymptotic critical value, then so does $(A z)_1$. Thus, this element of randomisation removes the
necessity of choosing $s$.

\begin{lemma}\label{lem:ksat}
  Let $\f \in \K[\bz]^p$ be a dominant polynomial mapping. Let
  $K_\infty(\f)$ be the set of asymptotic critical values of $\f$ and
  $K_s^j(\f^A)$ be defined as in Notation~\ref{not:ksat}. There exists a non-empty Zariski open subset
  $\cOA$ of $\GL_n(\K)$ such that for $A \in \cOA$ the following equality
  holds:
   \[K_\infty(\f) \subseteq \bigcup_{j=1}^p K_1^j(\f^A).\]
\end{lemma}

\begin{proof}[Proof of Lemma~\ref{lem:ksat}]
  Suppose $c \in K_\infty(\f)$. Then, there exists some sequence
  $(\bx_t)_{t\in \N}\in (\C^n)^{\N}$ such that as
  $t \to \infty$, \[\norm{\bx_t}\to\infty, \f(\bx_t)\to c \text{ and } \|\bx_t\|\kappa(\d\f(\bx_t))\to 0. \]
  Then, by considering the real and complex parts of the latter two
  limits, one defines a finite number of polynomials with real coefficients that give constraints
  defining a disc centred at $c$ and $0$ respectively, a semi algebraic set in $\R^{2n}$ by the
  isomorphism between $\C^n$ and $\R^{2n}$.
  Therefore, by the curve selection lemma at infinity~\cite[Lemma 3.3]{KOS}, which
  is obtained from a semialgebraic compactification of $\R^{2n}$ and the classical curve selection
  lemma~\cite[Theorem 2.5.5]{BCR}, and by~\cite[Proposition 8.1.12]{BCR} there exists a Nash curve
  $\gamma: (0,1) \to \R^{2n}$ such that
  \[
    \f(\gamma(t)) \to c, \norm{\gamma(t)} \to \infty \text{ and } \norm{\gamma(t)}\kappa(\d \f(\gamma(t))) \to 0 \text{ as } t \to 0.
  \]
  In the case $\K = \C$, one may then consider the Nash curve, a semialgebraic curve in the class
  $C^\infty$ defined from $(0,1)$ to $\C^n$. In either case, since $\gamma$ is a Nash mapping, we can
  express each component of $\gamma$ as a Puiseux series in $t$ by a Taylor expansion at $0$.

  We denote this expansion $z(t)$.
  In this way, each component of the expansion of $\gamma$ has
  finitely many terms with negative exponents. 
  In particular, the  order of each component, the least value $r$ such
  that the coefficient of the term $t^r$ is 
  non-zero, corresponds to the dominant term in the limit $t \to 0$.
  
  Let $\lambda_i$ denote the coefficient of the term of
  $z_i(t)$ with exponent ${\ord(z_i(t))}$. Denote by $\mathscr{I}$ the index set of all combinations of
  the elements of the set $\{1, \dots, n\}$. Then, we consider the finite set of linear equations of
  the $\lambda_i$:
  \[
    \mathscr{V} = \left\{\sum_{j=1}^{|\phi|} a_{\phi_j} \lambda_{\phi_j} \; \middle| \; a \in \K^n, \phi \in \mathscr{I}\right\}.
  \]
  Consider the group of $n \times n$ invertible matrices $\GL_n(\K)$ with entries in $\K$. Then, the set
  $\mathscr{V}$ gives the set of equations such that if any row of a matrix $T \in \GL_n(\K)$ satisfies
  any of these equations, there is a cancellation in the term of
  highest degree in the product $T z(t)$.
  The zero set of each equation in $\mathscr{V}$ defines a proper Zariski closed subset of $\GL_n(\K)$.
  Clearly, for any equation, these Zariski closed subsets are not dense. Therefore, since there are
  finitely many equations in $\mathscr{V}$, the union of the zero sets of all equations in $\mathscr{V}$
  defines a proper Zariski closed subset of $\GL_n(\K)$. Thus, there exists a non-empty Zariski open
  subset $\cOA^{-1}$ of $\GL_n(\K)$ such that no cancellation occurs in the highest degree so that all
  components of the product $T z(t)$ grow at the same speed with $t$. Since $\cOA^{-1} \subset \GL_n(\K)$,
  we may consider the proper Zariski closed subset of $\GL_n(\K)$, $\cOA$, defined by
  \[ A \in \cOA \iff A^{-1} \in \cOA^{-1}. \]

  Thus, choose $A \in \cOA$. Consider the polynomial $\f^{A} = \f(A z)$ and the path
  $y(t) = A^{-1}z(t)$. Then, as $t \to \infty$ we have $\norm{y(t)} \to \infty$ and $\f^{A}(y(t)) \to
  c$. In particular, since there is no cancellation by the choice of $A$, we have that $y_1(t) \to
  \infty$. Furthermore, with the Rabier distance $\nu$ measuring the distance to the space of
  singular operators we have that $\norm{z(t)}\nu(\d \f(z(t))) \to 0.$
  By the genericity of $A$, we have
  $\norm{y(t)} \nu(\d \f^{A}(y(t))) \to 0.$ Therefore, by~\cite[Corollary 2.1]{KOS}, we have
  \[\norm{y(t)} \kappa(\d \f^{A}(y(t))) \to 0.\]
  Now, choose $j$ such that $\kappa(\d \f^{A}(y(t))) = \norm{w_j(y(t))}$ where $w_j$ is the
  restriction of $\d f_j^{A}$ to the kernel of $\jac(\f^{A[j]})(z)$.
  
  Since $\G_{n-k+1}(\K^n \times \K)$ is compact, by~\cite[Lemma
  5.1]{milnorstasheff}, there is a limit $W_s^j$ of graphs
  $y_1(t) w_j(y(t))$. Thus, by
  $\norm{y(t)}\norm{w_j(z(t))} \rightarrow 0$, $W_s^j \in \Lambda$
  so that $(0, c, W_s^j) \in L_s^j(\f)$ and so $c \in K_s^j(\f^{A})$. Therefore,
  \[
    K_\infty(\f) = K_\infty(\f^A) \subseteq \bigcup_{(s,j) =
      (1,1)}^{(n,p)} K_s^j(\f^A) = \bigcup_{j=1}^p K_1^j(\f^A).\qedhere
  \]
\end{proof}

In summary, the construction given in Notation~\ref{not:ksat} begets an algorithm to compute
the asymptotic critical values of a dominant polynomial mapping. Moreover, Lemma~\ref{lem:ksat}
allows us to reduce the number of sets we must compute from $n p$ to just $p$. We make further
improvements and describe the following algorithms in Section~\ref{sec:impro}. 


%% file: newfaster_algo.tex
\section{Algorithms}
\label{sec:impro}

In this section, we give a geometric result that allows us
to introduce an additional element of
randomisation. By next translating the geometric objects defined in Section~\ref{sec:theorem} into
an algebraic setting, we give an algebraic proof that the set of asymptotic critical values has
codimension at least $1$. Furthermore, we define an algorithm of algebraic elimination for computing
a finite list of polynomials whose zero set contains the asymptotic critical values of the input
dominant polynomial mapping. With the geometric result, this algorithm only needs to introduce $p+1$
new indeterminates, rather than introducing $n+1$ indeterminates as in the algorithm one derives from
Lemma~\ref{lem:ksat}.

In addition to this, by making the relation to some incidence varieties, we reduce the number of
introduced variables to just $p$, one for each value of our polynomial mapping. These reductions
undeniably come with great complexity improvements and more efficient
algorithms, particularly in the special case where $p=1$. This second algorithm also provides
an avenue to tighten the degree bound on
the set of asymptotic critical values, which is investigated in Section~\ref{sec:degree}. The
complexity of the two algorithms described in the section is analysed in Section~\ref{sec:comp}. 

\subsection{Geometric result}
We start with a geometric proposition involving the following objects:

Let $\Gamma$ be a subspace of $\K^N$ of dimension $\theta$, for some $\theta < N$. Then, define the
set $\mathcal{E} \subset \G_{N-\theta+1}(\K^N)$ to be the subset of the Grassmannian of subspaces of
dimension $N - \theta + 1$ such that the projection of every subspace $E \in \mathcal{E}$ onto
$\Gamma$ has dimension $1$. 

\begin{proposition}\label{prop:e}
  Let $W$ be a set of dimension $\alpha \geq \theta$ so that the Zariski closure of
  $W$, $V = \overline{W}$, is equidimensional in an ambient space of dimension $N$ and let $\Gamma$ and
  $\mathcal{E}$ be defined as above. Suppose there exists a hypersurface $Z$ such that
  $Z \cap W = \emptyset$ and that $V \setminus Z = W$. Suppose the projection of $V$ onto $\Gamma$ has
  dimension $\theta$. Then, there exists a proper non-empty Zariski open subset $\cOE \subset
  \mathcal{E}$ so that for $E \in \cOE$, $\overline{W \cap E} = V \cap E$.
\end{proposition}

\begin{proof}[Proof of Proposition~\ref{prop:e}]
  For any subspace $E$, $V \cap E$ is an algebraic set containing $W \cap E$.
  Therefore, we have that $\overline{W \cap E} \subset V \cap E$ and so
  it remains to show that $V \cap E \subset \overline{W \cap E}$. 
  
  By definition, the projection of $V$ onto the subspace $\Gamma$ is dominant. Then, by Thom's
  transversality theorem~\cite[page~67]{GVP} and by the definition of the set $\mathcal{E}$,
  a generic element $E$ of $\mathcal{E}$ intersects $V$ transversally.
  Thus, there exists a proper non-empty Zariski open subset
  $\cO_1 \subset \mathcal{E}$ so that for all $E \in \cO_1$, the intersection of $V$
  and $E$ is transverse. Since $V$ is an equidimensional variety, by the genericity of $E$, $V\cap E$
  is also equidimensional. By~\cite[Theorem 1.24]{IS}, the intersection $V \cap E$ has dimension
  $\alpha - \theta + 1$ which is at least
  $1$. Let $F=Z\cap V$, then $F$ has codimension $1$ in $V$.
  By another application of Thom's transversality theorem~\cite[page~67]{GVP}, there exists a 
  proper non-empty Zariski open subset $\cO_2 \subset \mathcal{E}$ so that for all $E \in \cO_2$,
  the intersection of $E$ and $F$ is transverse. Let $\cOE = \cO_1 \cap \cO_2$, a proper non-empty
  Zariski open subset of $\mathcal{E}$. Therefore, for all $E \in \cOE$, $E \cap F$ has dimension
  $\alpha - \theta$.
  Now, let $E \in \cOE$. Then, for all of the finitely many irreducible components $U$ of $V \cap E$ we
  want to show that $U \subset \overline{W \cap E}$. Let $U$ be one such irreducible component.
  Note that the dimension of $U$ is $\alpha - \theta + 1$ since $V \cap E$ is equidimensional.
  Furthermore, we have $\dim(U \cap F) < \alpha - \theta + 1$. This implies
  that $\overline{U \setminus F} = U$. Combining this with \[U \setminus F \subset
    (V \cap E) \setminus F = W \cap E,\] we find
  by taking the Zariski closure that $U \subset \overline{W\cap E}$. Since this holds for all
  irreducible components of $V \cap E$, we conclude that $V \cap E \subset \overline{W \cap E}$.
\end{proof}

\subsection{Algebraic description of asymptotic critical values}\label{sec:alg}

In this subsection we translate the objects $M_s^j(\f)$, $L_s^j(\f)$ and $K_s^j(\f)$, defined as in
Notation~\ref{not:ksat}, to an algebraic setting that will allow the use of algebraic elimination
algorithms. To this end, we must derive polynomials from which we can give varieties that are equal to
these sets.

Let $\f = (f_1, \dots, f_p) \in \K[\bz]^p$ be a dominant polynomial mapping. Let
$\cOA \subset \GL_n(\K)$ be a non-empty Zariski open subset such that any $A \in \cOA$
satisfies the genericity requirements of Lemma~\ref{lem:ksat}.

For each choice of $j$ we aim to compute a representation of $K_1^j(\f^A)$.
Let $\Lambda = \G_{n-p+1}(\C^n \times 0)$ and recall the definition of $K_1^j(\f^A)$,
\[
  K_1^j(\f^A) = \pi\left(L_1^j(\f^A)\right) = \pi\left(\overline{\gr M_1^j(\f^A)} \cap (\{z \in \C^n | z_1 = 0\} \times \C^p \times \Lambda)\right).
\]
Consequently, the first step must be to compute a representation of
the graph of $M_1^j(\f^A)$. Let $W_1^j(z)$ denote the graph of the
map defined by $z \mapsto z_1 w_j(z)$, then $M_1^j(\f^A)$ is defined by
\[
  M_1^j(\f^A) = (\f^A(\tau_1(z)), W_1^j(\tau_1(z))).
\]
Denote by $\mathcal{E}$ the set of $(n+p+1)$-dimensional subspaces that are defined by the set of
equations $u_1 - r_1 e = \dots = u_{n-p+1} - r_{n-p+1} e = 0$, where $e$ is an indeterminate, for
$r_1, \dots, r_{n-p+1} \in \K$. Additionally, denote by $\br$ the numbers $r_1, \dots, r_{n-p+1}$.
  
\begin{corollary}\label{cor:msj}
  Let $\f = (f_1, \dots, f_p) \in \K[\bz]^p$ be a dominant polynomial mapping.
  Let $M_1^j(\f^A)$ be defined as in Notation~\ref{not:ksat} with $A$ chosen to satisfy the genericity
  assumption of Lemma~\ref{lem:ksat}. Let $\br \in \K$ be chosen so that
  the subspace $E\in\mathcal{E}$ they define satisfies the genericity condition of
  Proposition~\ref{prop:e}.
  Then, there exists a polynomial tuple $(g_1, \dots, g_{n+1})$ and
  polynomial $h$ with entries in the polynomial ring
  $\K[\bz, \bc, e]$ such that,
  \[
    M_1^j(\f^A) = \overline{\V(g_1, \dots, g_{n+1}) \setminus \V(h)},
  \]
  \[
    L_1^j(\f^A) = \overline{\V(g_1, \dots, g_{n+1}) \setminus \V(h)} \cap
    \V(z_1, e).
  \]
  Furthermore, the dimension of $\V(g_1, \dots, g_{n+1})$ is $p$.
\end{corollary}

\begin{proof}[Proof of Corollary~\ref{cor:msj}]
  Since $\f$ is a dominant polynomial mapping, it is clear that $\f^A$ is also dominant.
  Thus, $M_1^j(\f^A)$ is well-defined outside of a nowhere dense algebraic set. We must derive
  equations defining this set in the algorithm. Firstly, we must remove the set $\V(z_1)$.
  Furthermore, we require that $W_1^j(\tau_1(z))$ be of
  dimension $n-p+1$ to be an element of the Grassmannian $\G_{n-p+1}(\C^n \times \C)$. 
  Recall that $w_j(z)$ is the restriction of $\d f_j^A$ to the kernel of the Jacobian matrix of $\f$
  with the $j$th row removed. Then, excluding $z_1 = 0$, $M_1^j(\f^A)$ is well-defined so long as the
  determinant of the submatrix given by the first $p-1$ columns of the Jacobian matrix
  $\jac((\f^A)^{[j]})$ is not $0$. We denote this determinant $\delta(z)$. 

  Next, we compute the Jacobian matrix $\jac((\f^A)^{[j]})$ and a basis $B$ for its null space.
  We can accomplish this by evaluation interpolation techniques using a Kronecker substitution to
  reduce the problem to the univariate case. The details of this, along with the complexity analysis,
  is given in Section~\ref{sec:comp}.

  Since $\jac((\f^A)^{[j]})$ has rank $p-1$, the basis $B$ consists of $n-p+1$ vectors, each with $n$
  rational functions as entries. The denominators of these functions describe an algebraic set where
  the function $M_s^j(f^A)$
  is not defined. By~\cite{gencramerrule}, $\delta(z)$ is the common denominator of these functions.
  Define $v_1(z), \dots, v_{n-p+1}(z)$ to be such that $v_i(z)$ is the product of the gradient of
  $f_j^A$ with the $i$th element of the basis $B$. Let $\lambda(t) \in \K^n$ be a path such that
  \[\lambda_1(t) v_1(\lambda(t)) = \dots = \lambda_1(t) v_{n-p+1}(\lambda(t)) \to 0 \text{ as } t \to \infty\]
  then, we have that $W_1^j(\lambda(t)) \to W$, for some $W \in \Lambda$. Then, apply the
  transformation $\tau_1$ to obtain the rational mappings
  $(\tau_1(z))_1 v_1(\tau_1(z)), \dots, (\tau_1(z))_1 v_{n-p+1}((\tau_1(z))_1)$.
  As discussed above, this is well-defined for
  $z \in \K^n \setminus \V(z_1\numer(\delta(\tau_1(z))))$.
  Consider the graph, $\mathscr{G}$, of these rational mappings with value variables
  $u_1, \dots, u_{n-p+1}$. We see that the intersection with the variety $\V(u_1, \dots, u_{n-p+1})$
  captures the points where the linear map $w_j(z)$ is identically the zero map. Therefore, this gives an algebraic
  version of the intersection with the space $\Lambda$.
  
  We can now define the polynomial tuple that gives our representation of the graph
  of $M_1^j(\f^A)$. To do so, we introduce independent variables $\bc$ and
  $u_1, \dots, u_{n-p+1}$. Define the polynomial tuple,
  \[
    N = (f_1^A - c_1, \dots, f_p^A - c_p, z_1 v_1 - u_1, \dots, z_1 v_{n-p+1} - u_{n-p+1}).
  \]
  Then, applying the transformation $\tau_1$ and taking the numerators of the resulting functions we
  have,
  \[
    \overline{\gr M_1^j(\f^A)} = \overline{\V(\numer(N(\tau_1(z)))) \setminus \V(z_1\numer(\delta(\tau_1(z))))}.
  \]

  From the above, it is now clear that the intersection with $\Lambda$ is accomplished by intersecting
  with the variety $\V(u_1, \dots, u_{n-p+1})$. Therefore,
  \[ L_1^j(\f^A)=\overline{\V(\numer(N(\tau_1(z)))) \setminus \V(z_1\numer(\delta(\tau_1(z))))} \cap \V(z_1, u_1, \dots, u_{n-p+1}). \]  
  
  It is clear that for all $E$ where $\br$ are all
  non-zero  we have \[\V(u_1, \dots, u_{n-p+1}) = E \cap \V(e).\] Hence, we may write: 
  \[ L_1^j(\f^A)=\overline{\V(\numer(N(\tau_1(z)))) \setminus \V(z_1\numer(\delta(\tau_1(z))))} \cap
    E \cap \V(z_1, e). \] 
   
  Note that $\mathcal{E}$ satisfies the conditions of Proposition~\ref{prop:e} with $\Gamma$ as the
  $(u_1, \dots, u_{n-p+1})$-space and recall that $E$ satisfies the genericity condition. 
  Furthermore, the Zariski closure of the graph of $M_1^j(\f^A)$ is an
  equidimensional variety of dimension $n$ in an ambient space of dimension $2n+1$ and is such that
  \[
    \overline{\gr M_1^j(\f^A)} \setminus \V(z_1\numer(\delta(\tau_1(z)))) = \gr M_1^j(\f^A).  
  \]
  Since $\f$ is a dominant mapping, so is $\f^A$. Therefore, the projection of the Zariski
  closure of the graph of $M_1^j$ onto the $(u_1, \dots, u_{n-p+1})$-space is dominant. Thus, we may
  apply Proposition~\ref{prop:e}:
  \[
    \overline{\gr M_1^j(\f^A) \cap E} = \overline{\gr M_1^j(\f^A)} \cap E.
  \]
  Furthermore, the intersection $\overline{\gr M_1^j(\f^A)} \cap E$ has dimension $p$.
  By~\cite[Theorem 4.3.4]{CLO}, since the polynomials \[
    u_1 - r_1e, \dots, u_{n-p+1} - r_{n-p+1}e
    \]are unaffected by the transformation $\tau_1$,
  for $1 \leq k \leq n-p+1$ we may replace $u_k$ by $r_k e$. Thus, we arrive at the tuple
  \[
    \tilde{N} = (f_1^A - c_1, \dots, f_p^A - c_p, z_1 v_1(z) - r_1 e, \dots, z_1 v_{n-p+1} - r_{n-p+1} e),
  \]
  so that
  \[
    \overline{\gr M_1^j(\f^A)} = \overline{\V(\numer(\tilde{N}(\tau_1(z)))) \setminus \V(z_1\numer(\delta(\tau_1(z))))}.
  \]
  We find that the polynomial tuple $\tilde{N}$ and polynomial $z_1\numer(\delta(\tau_1(z)))$ satisfy
  the statement.  
\end{proof}

Now that polynomials whose zero set contains the
asymptotic critical values have been described, we may give
an algebraic proof that the set of asymptotic critical values is contained in a set of codimension
at least one. To do this, we will show that the sets we aim to derive in the algorithms in this section
have codimension at least one. A critical fact meaning that the algorithms output will be a finite
list of polynomials.

\begin{corollary}\label{cor:dim}
  Let $\f = (f_1, \dots, f_p) \in \K[\bz]^p$ be a dominant polynomial mapping. Let $1 \leq j \leq p$
  and let $K_1^j(\f) \in \C^p$ be defined as above. Then, $\dim(K_1^j(\f)) \leq p-1$. Furthermore,
  $K_\infty(\f)$ has codimension at least $1$ in $\C^p$.
\end{corollary}

\begin{proof}[Proof of Corollary~\ref{cor:dim}]
  By Corollary~\ref{cor:msj}, there exists a tuple of polynomials $(g_1, \dots, g_{n+1}) \in \K[\bz,\bc,e]^{n+1}$ and polynomial $h \in \K[\bz, \bc, e]$ such that for $A\in\GL_n(\K)$
  satisfying the genericity condition of Lemma~\ref{lem:ksat},
  \[
    \overline{\gr M_1^j(\f^A)} = \overline{\V(g_1, \dots, g_{n+1}) \setminus \V(h)},
  \]
  \[
    L_1^j(\f^A) = \overline{\V(g_1, \dots, g_{n+1}) \setminus \V(h)} \cap
    \V(z_1, e),
  \]
  where $\V(g_1, \dots, g_{n+1})$ has dimension $p$. Furthermore, by the proof of
  Corollary~\ref{cor:msj}, $z_1$ is a factor of $h$. Therefore, $z_1$ is not identically
  zero at any irreducible component of $\overline{\V(g_1, \dots, g_{n+1}) \setminus \V(h)}$.
  By~\cite[Theorem 1.24]{IS}, the intersection with $\V(z_1)$ to derive $L_1^j(\f^A)$ necessarily
  reduces the dimension by $1$. Hence, $\dim(L_1^j(\f^A)) \leq p-1$. Recall that
  $K_1^j(\f) = \pi(L_1^j(\f))$ where $\pi$ is the projection map onto the $\bc$-space.
  Since the projection cannot increase the dimension~\cite[Theorem 1.25]{IS},
  $\dim(K_1^j(\f)) \leq p-1$.

  Furthermore, by Lemma~\ref{lem:ksat} for such a matrix $A\in\GL_n(\K)$ we have,
  \[K_\infty(\f) \subseteq \bigcup_{j=1}^p K_1^j(\f^A).\]
  Since the above holds for all $j$, we have that
  the finite union $\bigcup_{j=1}^p K_1^j(\f^A)$ has dimension at most $p-1$. Thus,
  $K_\infty(\f)$ has codimension at least $1$ in $\C^p$.
\end{proof}

The algorithms described in this section use the above results and constructions
to compute the asymptotic critical values of a dominant polynomial mapping using
algebraic methods. To present the following algorithms we introduce some functions
and subroutines that will feature in our algorithms. 

\subsection{Subroutines}\label{sec:subroutines}

We introduce 3 subroutines that will be used across all the algorithms
featured in this article.

\begin{description}
\item[$\Eliminate(P,\bv,\bw)$:]\ 
  \begin{description}
  \item[Input:] \textsf{$P$, a finite basis of an ideal, $I$, of a polynomial
    ring (with base field $\K$ and two lists of indeterminates, $\bv$
    and $\bw$) which we denote $\K[\bv, \bw]$.}
  \item[Output:] \textsf{$E$, a finite basis of the ideal $I \cap \K[\bw]$}.
  \end{description}

\item[$\Intersect(P_1, \ldots, P_k)$:]\
  \begin{description}
  \item[Input:] \textsf{$P_1, \dots, P_k$, finite bases of ideals,
    $I_1, \dots, I_k$, of a polynomial ring.}
  \item[Output:] \textsf{$P$, a finite basis of the ideal
    $\bigcap_{i=1}^k I_i$.}
  \end{description}

\item[$\Saturate(P_1,P_2)$:]\
  \begin{description}
  \item[Input:] \textsf{$P_1,P_2$, finite bases of ideals, $I_1, I_2$, of a
    polynomial ring.}
  \item[Output:] \textsf{$S$, a finite basis of the ideal $I_1:I_2^\infty$}.
  \end{description}
\end{description}

\begin{remark}\label{rem:sr}
  We remark that algorithms for these subroutines exist, in particular all can be accomplished using
  \gbs. We refer to~\cite[page~122]{CLO},~\cite[Proposition 6.19]{beckerGB} and~\cite{Eisenbud,Bayer}
  for algorithms for computing a finite basis for respectively elimination ideals, intersection of
  ideals and the saturation of ideals.
\end{remark}

\subsection{First algorithm}\label{ss:first}

\begin{algorithm2e}[htbp!]
   \DontPrintSemicolon
  \nonl\TitleOfAlgo{$acv1$}
    \label{algo:acv_e}
  \KwIn{$\f:\K^n \rightarrow \K^p$ a dominant polynomial mapping with components in the ring $\K[\bz]$,
    the list $\bz$.}
  \KwOut{$\ACV$, a finite list of polynomials whose zero set has codimension at least $1$ in $\C^p$
    and contains the set of asymptotic critical values of $\f$.}  
  Generate a random change of variables $A\in\K^{n\times n}$ and set\;
  \nonl$\f^A\leftarrow \f(A z)$.\;
  \For{$j$ \KwFrom $1$ \KwTo $p$}{
    Generate random numbers $\br \in \K$.\;
    $B\leftarrow\text{Basis of the kernel of }\jac((\f^A)^{[j]})$.\;
    $(v_1(z), \dots, v_{n-p+1}(z))\leftarrow\d f^A_j B$.\;
    $\delta(z) \leftarrow$ the determinant of the first $p-1$ columns of $\jac((\f^A)^{[j]})$.\; 
    $N(z)\leftarrow\{f^A_1(z)-c_1,\ldots,f^A_p(z)-c_p,
    z_1v_1(z)-r_1e,\ldots,
    z_1v_{n-p+1}(z)-r_{n-p+1}e\}$.\;
    $G \leftarrow \numer(N(\tau_1(z)))$.\;
    $G_s \leftarrow\Saturate(G, z_1\numer(\delta(\tau_1(z))))$.\;
    $L \leftarrow G_s \cup \{z_1, e\}$.\;
    $V_j \leftarrow\Eliminate(L, \{\bz, e\}, \{\bc\})$.\;
  }
  $\ACV\leftarrow\Intersect(V_1, \dots, V_p)$.\;
  \KwRet $\ACV$.\;
\end{algorithm2e}

We first define the objects that will be crucial in the proof of correctness and termination
of Algorithm~\ref{algo:acv_e}.

\begin{theorem}\label{thm:algo2}
  Let $\f = (f_1, \dots, f_p) \in \K[\bz]^p$ be a dominant polynomial mapping.
  Suppose that $A \in \GL_n(\K)$ satisfies the genericity condition of Lemma~\ref{lem:ksat} and that
  $\br$ and the corresponding subspace $E \in \mathcal{E}$ satisfies the genericity
  condition of Proposition~\ref{prop:e}. Then, Algorithm~\ref{algo:acv_e} terminates and returns as
  output a finite basis whose zero set has codimension at least $1$ in $\C^p$ and contains the set of
  asymptotic critical values of~$\f$.
\end{theorem}

\begin{proof}[Proof of Theorem~\ref{thm:algo2}]
  Firstly, Algorithm~\ref{algo:acv_e} uses linear algebra and, as in Remark~\ref{rem:sr},
  multivariate polynomial routines that are correct and terminate. Hence,
  Algorithm~\ref{algo:acv_e} terminates in finitely many steps. 

  By Lemma~\ref{lem:ksat}, for such a matrix $A \in \GL_n(\K)$,
  \[K_\infty(\f) \subseteq \bigcup_{j=1}^p K_1^j(\f^A).\]
  By Corollary~\ref{cor:dim}, for each $1 \leq j \leq p$, $K_1^j(\f^A)$ has codimension at least $1$ in 
  $\C^p$. Thus, the finite union $\bigcup_{j=1}^p K_1^j(\f^A)$ has codimension at least $1$ in $\C^p$.
  The goal is then to compute $p$ finite sets of polynomials such that their zero sets are
  $\overline{K_1^1(\f^A)}, \dots, \overline{K_1^p(\f^A)}$. The union of these zero sets would then
  contain the the asymptotic critical values of $\f$. Hence, choose $1 \leq j \leq p$. Then, generate
  random numbers $\br \in \K$. 

  By Corollary~\ref{cor:msj}, we derive polynomial tuple $(g_1, \dots, g_{n+1})$ and polynomial
  $h$  with entries in the polynomial ring
  $\K[\bz, \bc, e]$ such that,
  \[
    \overline{\gr M_1^j(\f^A)} = \overline{\V(g_1, \dots, g_{n+1}) \setminus \V(h)},
  \]
  \[
    L_1^j(\f^A) = \overline{\V(g_1, \dots, g_{n+1}) \setminus \V(h)} \cap \V(z_1, e).
  \]

  As in the proof of this corollary, to compute the polynomial tuple $(g_1, \dots, g_{n+1})$ and
  polynomial $h$, we first compute a basis, $B$, of the kernel of $\jac((\f^A)^{[j]})$. This is
  accomplished by a evaluation interpolation method detailed in Section~\ref{sec:comp}. Then,
  it is easy to see that $h = z_1\numer(\delta(\tau_1(z)))$, where $\delta$ is computed by step $6$ of
  Algorithm~\ref{algo:acv_e} and the tuple $(g_1, \dots, g_{n+1}) = G$, where $G$ computed by step $8$
  of Algorithm~\ref{algo:acv_e}. 

  The next stage is to compute the Zariski closure of the graph of $M_1^j(\f^A)$. Thus, we must compute
  a finite list of polynomials whose zero set is
  \[\overline{\V(g_1, \dots, g_{n+1}) \setminus \V(h)}.\]
  By~\cite[Chapter~4, Section~4, Theorem~10]{CLO}, we may do this through saturations. Therefore, we
  apply the subroutine $\Saturate$ to $G$ to saturate by the ideal
  $\ideal{z_1\numer(\delta(\tau_1(z)))}$. We denote by $G_s$ the finite list of polynomials that is
  returned by the $\Saturate$ subroutine and conclude:
  \[ \overline{\gr M_1^j(\f^A)} = \V(G_s). \]

  Thus, as in Corollary~\ref{cor:msj}, we compute $L_1^j(\f^A)$ by intersecting with the
  variety $\V(z_1, e)$. By~\cite[Chapter~4, Section~3, Theorem~4]{CLO}, we add the polynomials $z_1, e$
  to the list $G_s$ to define the finite list of polynomials $L$ so that
  \[ \V(L) = \V(G_s) \cap \V(z_1, e). \] It remains to project onto the $\bc$-space.
  By~\cite[Chapter~4, Section~4, Theorem~4]{CLO}, we apply the subroutine $\Eliminate$ to the list $L$
  in order to eliminate all variables except $c_1, \dots, c_p$. The result is a finite list of
  polynomials, $V_j$, whose zero set is the set $\overline{K_1^j(\f^A)}$ which contains the set
  $K_1^j(\f^A)$ by definition. Thus, by Corollary~\ref{cor:dim}, the algebraic set $\V(V_j)$ has
  codimension at least $1$ in $\C^p$. 
  
  We perform these steps for all $j$ from $1$ to $p$ to obtain $V_1, \dots, V_p$.
  By~\cite[Chapter~4, Section~3, Theorem~15]{CLO}, the final step of computing their union can be
  performed by applying the subroutine $\Intersect$ to the lists $V_1, \dots, V_p$. The output is
  a finite list of polynomials which we denote $\ACV$. We conclude that \[
    K_\infty(\f) \subseteq \bigcup_{j=1}^p K_1^j(\f^A) \subseteq \V(\ACV). \qedhere\] 
\end{proof}
                  
\begin{example}
   We give an example of how to use the computer
  algebra system \Maple to implement Algorithm~\ref{algo:acv_e} with \FGb~\cite{FGb}, implemented in C,
  to perform the \gb computations. We shall report solely the inputs as the outputs are impractical to
  give here. The details and the results of these computations can instead be found on the webpage:
  \emph{\url{https://www-polsys.lip6.fr/~ferguson/globalacv.html}}
  
  Let $\g =  \left(z_1z_2,z_1z_3\right)$. We give the calls required to compute the set $K_1^1(\g^A)$
  for some invertible matrix $A$ satisfying the genericity condition in Theorem~\ref{algo:acv_e}.
  The case $j=2$ follows similarly.

  We begin by generating a random seed so we can generate a random matrix $A$ and define
  the polynomial \verb+gA+:

\begin{verbatim}
 > randomize();
 > A := LinearAlgebra:-RandomMatrix(3);
 > vars := [z1,z2,z3];
 > Asubs := {seq(vars[i] = add(A[i,j]*vars[j], j=1..3),i=1..3)};
 > gA := subs(Asubs, g);
 > jacgA := VectorCalculus:-Jacobian(gA, vars);
 > dgA := jacgA[1..1,1..3];
\end{verbatim}

  We first must compute the polynomial mapping $w_1(z) = (v_1(z), v_2(z))$. To do so, we compute a basis of the kernel of $\jac((\g^A)^{[1]})$. We can then define $\tau_1$ and $G$.
\begin{verbatim}
 > B := LinearAlgebra:-NullSpace(jacgA[2..2,1..3]);
 > v1 := add(B[1][i]*dgA[i]*vars[1], i=1..3);
 > v2 := add(B[2][i]*dgA[i]*vars[1], i=1..3);
 > tau := {z1 = 1/z1, z2 = z2/z1, z3 = z3/z1};
 > G := numer(subs(tau, [gA[1] - c1, gA[2] - c2, rand()*e - v1, 
           rand()*e - v2]));
\end{verbatim}
    Recall that $M_1^1(\g^A)$
  is only defined where $\jac(\g^A)$ is full rank. Therefore, we must remove the
  algebraic set defined by the minors of this matrix, as well as the
  variety $\V(z_1)$. By the choice of a generic matrix $A$, it suffices to remove the algebraic set
  defined by the minor
  given by the first $p-1$ columns of $\jac((\g^A)^{[1]})$.
\begin{verbatim}
 > delta := jacgA[2,1];
 > Gs := FGb:-fgb_gbasis_elim([op(G), t*numer(subs(tau,delta))*vars[1]-1]
            , 0, [t], [op(vars),e,c1,c2]);
 > V := FGb:-fgb_gbasis_elim([op(Gs), vars[1], e], 0, [op(vars),e], [c]);
\end{verbatim}

  After performing these computations, similar to what we saw in Example~\ref{ex:msj}, $K_1^1(\g^A) =
  \V(c_1,c_2)$. A similar computation for $j=2$ would reveal the same. Thus, $(0,0)$ is the only
  possible asymptotic critical value of $\g^A$ and therefore of $\g$.
\end{example}

\subsection{Improved algorithm}

We can implement the idea of Algorithm~\ref{algo:acv_e}
in a different way. For fixed $1 \leq j \leq p$, consider a basis $B$ of
the kernel of $\jac((\f^A)^{[j]})$ and define $(v_1, \dots, v_{n-p+1}) = \d f^A_j B$.
By~\cite{gencramerrule}, we may assume that the
$v_i(z)$ have common denominator $\delta(z)$, the determinant of the first $p-1$ columns of
$\jac((\f^A)^{[j]})$. Then, define the list of polynomials $G$ by
\[
  \numer(f^A_1(\tau_1(z))-c_1,\ldots,f^A_p(\tau_1(z))-c_p,
  v_1(\tau_1(z))-z_1r_1e,\ldots, v_{n-p+1}(\tau_1(z))-z_1r_{n-p+1}e),
\]
and denote $\Eliminate(G, e, \{\bz, \bc\})$ by $G^\prime$. 
We force the map \[z \mapsto (v_1(\tau_1(z)), \dots, v_{n-p+1}(\tau_1(z)))\]
to be parallel to a generic vector $\br \in \K^{n-p+1}$.
We did this before by introducing a variable $e$. Instead, let $M$ denote 
the ideal generated by the numerators of the minors of the following matrix
evaluated at $\tau_1(z)$
\[
  \begin{bmatrix}
    v_1 & \cdots & v_{n-p+1} \\
    r_1 & \cdots & r_{n-p+1}
  \end{bmatrix}.
\]
If the minors of this matrix are set to $0$, there is a rank deficiency. This means that
the two rows are parallel. In this setting, we would therefore not need to introduce a variable $e$ to
consider the linear subspace $E$ from Algorithm~\ref{algo:acv_e}, but instead, we include the minors
of this matrix. It is easy to see that the minors discussed are exactly what is obtained by
eliminating the introduced variable $e$ from the ideal given by the basis $G$. This is the content of
the following lemma.
  
\begin{lemma}\label{lem:minors}
  Let $\f = (f_1, \dots, f_p) \in \K[\bz]^p$ be a dominant polynomial mapping. Let $A \in \GL_n(\K)$
  and let $\br \in \K$ so that the genericity assumptions of Theorem~\ref{thm:algo2}
  hold. Let $G$, $M$ and $G^\prime$ be defined as above.
  Then, the following equality holds:
  \[
    \ideal{\numer(f^A_1(\tau_1(z)) - c_1, \dots, f^A_p(\tau_1(z)) -
      c_p)} + M = \ideal{G^\prime}.
  \]
\end{lemma}

\begin{proof}[Proof of Lemma~\ref{lem:minors}]
  We shall prove this by double inclusion.

  Firstly, the numerators of the polynomials
    $f^A_1 - c_1, \dots, f^A_p - c_p$ at $\tau_1(z)$ are elements of
    $G$ and $\K[\bz, \bc]$, so it remains to
    show that the numerator of each minor is an element of
    $\ideal{G^\prime}$. Let $r_k v_i(\tau_1(z)) - r_i d_k(\tau_1(z))$
    be one such minor. Then,
    $r_i(r_k z_1e - d_k(\tau_1(z))) - r_k(r_i z_1e - v_i(\tau_1(z))) =
    r_k v_i(\tau_1(z)) - r_i d_k(\tau_1(z))$. Taking the numerators of
    both sides we find that
    $\numer(r_k v_i(\tau_1(z)) - r_i d_k(\tau_1(z))) \in
    \ideal{G^{\prime}}$. Thus, \[\ideal{\numer(f^A_1(\tau_1(z)) - c_1, \dots, f^A_p(\tau_1(z)) -
      c_p)} + M \subseteq \ideal{G^{\prime}}.\]
    
    On the other hand, let $g \in \ideal{G^\prime}$. By the definition of $G^\prime$,
    $g \in \ideal{G}$ such that all $e$-terms are cancelled. That is,
    \begin{align*}
      g &=  \numer(h_1 (f^A_1(\tau_1(z)) - c_1) + \dots + h_p (f^A_p(\tau_1(z)) - c_p) + \\
        &\qquad h_{p+1} (e z_1 r_1 - v_1(\tau_1(z))) + \dots + h_{n+1} (e z_1 r_{n-p+1} - v_{n-p+1}(\tau_1(z)))),
    \end{align*}
    so that $h_{p+1}, \dots, h_{n+1} \in \K[\bz, \bc, e]$
    are polynomials such that in the above sum, all terms involving $e$ sum to $0$. Consider the
    following monomial ordering, $e > z_1 > \dots > z_n > c_1 > \dots
    > c_p$. Then, the 
    leading term of each $\numer(e z_i r_i - v_i(\tau_1(z)))$ divides the leading term of the 
    polynomial  $e z_i^d \delta(z) r_i$, for some $d$ large enough. This leading term must cancel
    in the polynomial $g$ as it involves $e$. Therefore, 
    $(h_{p+1}, \dots, h_{n+1})$ is a syzygy on the leading terms of
    $\numer(e z_i r_i - v_i(\tau_1(z)))$. Recall that the $S$-polynomials generate the
    set of syzygies on the leading terms $e z_1^d r_1$~\cite[page~111]{CLO} and here the
    $S$-polynomials are simply the minors of the above matrix. It is therefore possible to
    rewrite $h_{p+1}, \dots, h_{n+1}$ as elements of $M$. Thus, \[ \ideal{G^\prime} \subseteq
    \ideal{\numer(f^A_1(\tau_1(z)) - c_1, \dots, f^A_p(\tau_1(z)) - c_p) + M}. \qedhere \]
\end{proof}

While this point of view allows us to drop the variable $e$
and thus reduce the number of
variables by $1$, it makes us introduce many more equations, namely
the $\binom{n}{2}$ minors of the matrix. The following lemma actually ensures that
only $n-p$ of them are needed.

\begin{lemma}\label{lem:subset}
  Let $f_1, \dots, f_k \in \K[x_1, \dots, x_n]$ and let
  $r_1, \dots, r_k \in \K$ with
  $r_1 \neq 0$. Consider the matrix $R$
  \[
    \begin{bmatrix}
      f_1 & \cdots & f_k \\      r_1 & \cdots & r_k
    \end{bmatrix}.
  \]
  Let $I \in \K[x_1, \dots, x_n]$ be the ideal generated by the
  minors of the matrix $R$. Then, $I = \ideal{r_2 f_1 - r_1
    f_2, \dots, r_k f_1 - r_1 f_k}$.
\end{lemma}

\begin{proof}[Proof of Lemma~\ref{lem:subset}]
  Clearly, $\ideal{r_2 f_1 - r_1 f_2, \dots, r_k f_1 - r_1 f_k}
  \subset I$. Let $M_{i,j} = r_j f_i - r_i f_j$ be a minor of the
  matrix $R$. Then, $r_j M_{1,i} - r_i M_{1,j} = r_1 r_i f_j - r_1
  r_j f_i = r_1 M_{i,j}$. Since $r_1 \neq 0$ we conclude that
  $M_{i,j} \in \ideal{r_2 f_1 - r_1 f_2, \dots, r_k f_1 - r_1
    f_k}$. This holds for all $1 \leq i < j \leq k$  and so $I
  \subset \ideal{r_2 f_1 - r_1 f_2, \dots, r_k f_1 - r_1 f_k}$. 
\end{proof} 

This approach with the minors leads us to the design of
Algorithm~\ref{algo:acv_minor}.

\begin{algorithm2e}[htbp!]
   \DontPrintSemicolon
  \nonl\TitleOfAlgo{$acv2$}
    \label{algo:acv_minor}
  \KwIn{$\f:\K^n \rightarrow \K^p$ a dominant polynomial mapping with components in the ring
    $\K[\bz]$, the list $\bz$.}
  \KwOut{$\ACV$, a finite list of polynomials whose zero set has codimension at least $1$ in $\C^p$
    and contains the set of asymptotic critical values of $\f$.}
  
  Generate a random change of variables $A\in\K^{n\times n}$ and set\;
  \nonl$\f^A\leftarrow \f(A z)$.\;
  \For{$j$ \KwFrom $1$ \KwTo $p$}{
    Generate random numbers $\br \in \K$.\;
    $B\leftarrow\text{Basis of the kernel of }\jac((\f^A)^{[j]})$.\;
    $(v_1(z), \dots, v_{n-p+1}(z))\leftarrow\d f^A_j B$.\;
    $\delta(z) \leftarrow$ the determinant of the first $p-1$ columns of $\jac((\f^A)^{[j]})$.\; 
    $N^\prime(z)\leftarrow\{f^A_1-c_1,\ldots,f^A_p-c_p,
    r_2v_1 - r_1v_2,\ldots,
    r_{n-p+1}v_1 - r_1v_{n-p+1}\}$.\;
    $G^\prime \leftarrow \numer(N^\prime(\tau_1(z)))$.\;
    $G_s^\prime \leftarrow\Saturate(G^\prime, z_1\numer(\delta(\tau_1(z))))$.\;
    $L^\prime \leftarrow G_s^\prime \cup \{z_1\}$.\;
    $V_j^\prime \leftarrow \Eliminate(L^\prime, \{\bz\}$, $\{\bc\})$.\;
  }
  $\ACV^\prime\leftarrow\Intersect(V_1^\prime, \dots, V_p^\prime)$.\;
  \KwRet $\ACV^\prime$.\;
\end{algorithm2e}

Since Algorithms~\ref{algo:acv_e}
and~\ref{algo:acv_minor} are quite similar, we want to be able to
reuse the proof of correctness of the former for the latter's.
 The difference between these two algorithms is the stage
of the algorithm where we eliminate the introduced variable $e$. In Algorithm~\ref{algo:acv_e}, this
is in step 12; for Algorithm~\ref{algo:acv_minor}, as in Lemma~\ref{lem:minors}, we consider an ideal
equal to the one given if we eliminated $e$ after step 9 in Algorithm~\ref{algo:acv_e}. However, the
steps in between involve a saturation with respect to an ideal that does not involve $e$. This
motivates the following lemma, which shall make use of the following notation. 

For fixed $1 \leq j \leq p$, consider a basis $B$ of the null space of $\jac((\f^A)^{[j]})$ and
define $v_1(z), \dots, v_{n-p+1}(z) = \d f^A_j B$. By~\cite{gencramerrule}, we may assume that
the $v_i(z)$ have common denominator $\delta(z)$, the determinant of the first $p-1$ columns of
$\jac((\f^A)^{[j]})$.
As in Algorithm~\ref{algo:acv_e}, define $G$ to be the list of polynomials
\[
 \{\numer(f^A_1(z)-c_1,\ldots,f^A_p(z)-c_p,
  z_1v_1(z)-r_1e,\ldots, z_1v_{n-p+1}(z)-r_{n-p+1}e)\},
\]
evaluated at $\tau_1(z)$. Define the following finite lists that are defined in
Algorithm~\ref{algo:acv_e} or Algorithm~\ref{algo:acv_minor}:
\begin{align*}
&  G_s = \Saturate(G, z_1\numer(\delta(\tau_1(z)))), \\
&  G^\prime = \Eliminate(G, e, \{\bz, \bc\}), \\
&  G_s^\prime = \Eliminate(G_s, e, \{\bz, \bc\}).
\end{align*}
The relationship between these lists is investigated in the following lemma.

\begin{lemma}\label{lem:elim}
  Let $\f = (f_1, \dots, f_p) \in \K[\bz]^p$ be a dominant polynomial mapping. Let $A \in \GL_n(\K)$
  and let $\br \in \K$ so that the genericity assumptions of Theorem~\ref{thm:algo2}
  hold. Let $G$, $G_s$, $G^\prime$ and $G_s^\prime$ be defined as above.
  Then, the following equality holds:
  \[\ideal{G_s^\prime} =\ideal{\Saturate(G^\prime, z_1\numer(\delta(\tau_1(z))))}.\]
\end{lemma}

\begin{proof}[Proof of Lemma~\ref{lem:elim}]
  By~\cite[Chapter~2, Section~7, Theorem~4]{CLO}, for a given ideal $I$ and term order $\geq$,
  there exists a unique reduced \gb of $I$ with respect to $\geq$. Note that given a \gb, which
  can be computed using Buchberger's algorithm~\cite[Chapter~2, Section~7, Theorem~2]{CLO}, there
  exists an algorithm to
  compute a reduced \gb~\cite[Proposition 5.56]{beckerGB}. Then, we shall prove this
  by considering algorithms which accomplish the subroutines $\Eliminate$ and $\Saturate$ by returning
  reduced \gbs. We will then show that the reduced \gbs returned by these
  algorithms, in this case, are the same. The key is that the $\Saturate$ subroutine
  can be performed using $\Eliminate$, and so we can perform both operations at the same time.

  Thus, we first make explicit, the algorithms we shall use. The key is that both $\Eliminate$ and
  $\Saturate$ can be performed by computing elimination ideals. To compute elimination ideals, for an
  ideal $I \subset \K[x_1, \dots, x_n]$, we compute a \gb with respect to a lexicographic monomial 
  ordering where $x_1 > \dots> x_n$. By~\cite[Chapter~3, Section~1, Theorem~2]{CLO}, removing from this
  \gb all
  polynomials that involve the variables $x_1, \dots, x_k$, for some $k < n$, gives a \gb of the ideal
  $I \cap \K[x_{k+1}, \dots, x_n]$. For $\Saturate$, by~\cite[Chapter~4, Section~4, Theorem~14]{CLO}, given an ideal
  $I \in \K[x_1, \dots, x_n]$ and a polynomial $g \in \K[x_1, \dots, x_n]$, we can compute a \gb of
  $I : \ideal{g}^\infty$ through the $\Eliminate$ subroutine. To do
  so, first compute a \gb $\cG$ of the
  ideal $I + \ideal{\ell g - 1}$, where $\ell$ is an independent variable, with respect to a
  lexicographic monomial ordering with $\ell > x_1 > \dots> x_n$. Then, $\cG \cap \K[x_1, \dots, x_n]$
  is a \gb of the ideal $I : \ideal{g}^\infty$ with respect to the ordering $x_1 > \dots > x_n$. The
  last operation here is the elimination of the variable $\ell$. As in~\cite[Chapter~3, Section~1, Theorem~2]{CLO}, that is simply remove from $B$ all polynomials that involve $\ell$. 

  Therefore, we may perform both saturations and the elimination of $e$
  with one \gb computation and an intersection. We simply add the polynomial \\
  $\ell z_1\numer(\delta(\tau_1(z))) - 1$ to the list $G$ and consider a \gb
  with respect to a lexicographic monomial ordering with $e > \ell > z_1 > \dots > z_n > c_1,
  \dots > c_p$. This is the output of $\Saturate(G^\prime, z_1\numer(\delta(\tau_1(z))))$.
  We could alternatively eliminate $e$ last, which would make the result $G_s^\prime$. 
  As both ways involve intersecting with the same polynomial ring, $\K[\bz, \bc]$,
  the result is the same.  
\end{proof}

We are now in a position to utilise Theorem~\ref{thm:algo2} to give a proof of correctness for Algorithm~\ref{algo:acv_minor}.

\begin{theorem}\label{thm:algo3}
  Let $\f = (f_1, \dots, f_p) \in \K[\bz]^p$ be a dominant polynomial mapping.
  Suppose that $A \in \GL_n(\K)$ satisfies the genericity condition of Lemma~\ref{lem:ksat} and that
  $\br$ and the corresponding subspace $E \in \mathcal{E}$ satisfies the genericity
  condition of Proposition~\ref{prop:e}. Then, Algorithm~\ref{algo:acv_minor} terminates and returns as
  output a finite basis whose zero set has codimension at least $1$ in $\C^p$ and contains the set of
  asymptotic critical values of~$\f$.
\end{theorem}

\begin{proof}[Proof of Theorem~\ref{thm:algo3}]
  We will show that the only difference between Algorithm~\ref{algo:acv_e} and
  Algorithm~\ref{algo:acv_minor} is when the independent variable $e$ is eliminated
  (and subsequently, $e$ is not added to the list $G_s$ in step 11). We then show that this does
  not lose any asymptotic critical values. The proof that the result set has codimension at least $1$
  in $\C^p$ is identical to the proof in Theorem~\ref{thm:algo2}.

  The first seven steps of Algorithm~\ref{algo:acv_e} and Algorithm~\ref{algo:acv_minor} are the same,
  where we choose $\br$ and $A$ to satisfy the genericity assumptions
  in Theorem~\ref{thm:algo2}.
  Then, in step~8 of Algorithm~\ref{algo:acv_minor}, the list of polynomials and rational mappings \[
    N^\prime(z) = \{f^A_1(z)-c_1,\ldots,f^A_p(z)-c_p,
    r_2v_1(z) - r_1v_2(z),\ldots,
    r_{n-p+1}v_1(z) - r_1v_{n-p+1}(z)\}
  \]
  is defined. By Lemma~\ref{lem:subset}, the ideal defined by this basis is equal to the following ideal
  \[\ideal{f^A_1(z)-c_1,\ldots,f^A_p(z)-c_p} + M.\]

  Denote by $G$ the following list of polynomials defined in step $8$ of Algorithm~\ref{algo:acv_e}
  \[
    \numer(f^A_1(\tau_1(z))-c_1,\ldots,f^A_p(\tau_1(z))-c_p,
    v_1(\tau_1(z))-r_1z_1e,\ldots, v_{n-p+1}(\tau_1(z))-r_{n-p+1}z_1e).
  \]

  By Lemma~\ref{lem:minors}, the ideal $\numer(\ideal{N^\prime(\tau_1(z))})$
  is generated by the basis returned by the subroutine $\Eliminate(G, e, \{\bz, \bc\})$.
  Therefore, we can conclude that the difference between Algorithm~\ref{algo:acv_e} and
  Algorithm~\ref{algo:acv_minor} is the step when the variable $e$ is eliminated. Then,
  Algorithm~\ref{algo:acv_minor} terminates by Theorem~\ref{thm:algo2}. It remains to show that
  the variety generated by the output is of dimension at most $p-1$ and contains 
  the set of asymptotic critical values of $\f$.

  Denote $\Saturate(G, z_1\numer(\delta(\tau_1(z))))$ by $G_s$ (the output of step 10
  of Algorithm~\ref{algo:acv_e}) and $\Eliminate(G_s, e, \{\bz, \bc\})$ by $G_s^\prime$.
  By Lemma~\ref{lem:elim}, $G_s^\prime$ is equal to
  \[\Saturate(\allowbreak \numer(N(\tau_1(z))), z_1\numer(\delta(\tau_1(z)))),\] the output of step 10
  of Algorithm~\ref{algo:acv_minor}.
  By~\cite[Chapter~3, Section~1, Theorem~2]{CLO}, the ideal generated by $G_s^\prime$ is contained in
  the ideal generated by $G_s$. Thus, by~\cite[Chapter~1, Section~4, Proposition~8]{CLO},
  $\V(G_s) \subset \V(G_s^\prime)$ and so
  $\V(G_s \cup \{z_1, e\})  \subset \V(G_s^\prime \cup \{z_1\})$. Since these two varieties are the
  varieties generated by the outputs of step 11 of Algorithm~\ref{algo:acv_e} and
  Algorithm~\ref{algo:acv_minor} respectively, and since the projection of $\V(G_s \cup \{z_1, e\})$
  onto the $\bc$-space is the set $K_1^j(\f^A)$ by Theorem~\ref{thm:algo2}, we can
  conclude that Algorithm~\ref{algo:acv_minor} returns as output a basis
  whose zero set contains the set of asymptotic critical values of $\f$. Furthermore, as in the proof
  of Theorem~\ref{thm:algo2}, the variety $\V(G_s)$ has dimension at most $p$ and therefore
  by~\cite[Theorem 1.24]{IS}, $\V(G_s) \cap \V(z_1)$ has dimension at most $p-1$. The result then
  follows from~\cite[Chapter~9, Section~4, Theorem~8]{CLO}, for varieties $X$ and $Y$,
  $\dim(X \cup Y) = \max(\dim(X), \dim(Y))$.
\end{proof}


%% file: proof_degree.tex
\section{Degree bounds and complexity estimates}\label{sec:proof}
In this section we prove the main results stated in Section~\ref{sec:intro}.
\subsection{Proof of Theorem~\ref{prop:degree}}
\label{sec:degree}

In this subsection, we use Algorithm~\ref{algo:acv_minor} to bound the degree of the set of asymptotic
critical values of polynomial mappings. In both of the algorithms designed in this paper, for each
choice of $1 \leq j \leq p$ assuming that $p>1$, we must compute a basis of the kernel of
$\jac((\f^A)^{[j]})$. Thus, we begin with a lemma bounding the degrees of the entries of such a basis.

\begin{lemma}\label{lem:basisdeg}
  Let $\f = (f_1, \ldots, f_p) \in \K[z_1, \ldots, z_n]^p$ be a dominant polynomial mapping and let $A \in
  \GL_n(\K)$ satisfy the genericity condition of Lemma~\ref{lem:ksat}. 
  Let $d = \max_{1 \leq i \leq p} \deg f_i$. Then, for all $1 \leq j \leq p$, there exists a basis $B$
  of the kernel of $\jac((\f^A)^{[j]})$ such that the entries of $B$ are rational functions whose
  numerators and denominators have degree at most $(p-1)(d-1)$.
\end{lemma}

\begin{proof}[Proof of Lemma~\ref{lem:basisdeg}]
  First, fix some $1 \leq j \leq p$ and consider the matrix $\jac((\f^A)^{[j]})$. Since $\f$ is
  dominant, by the genericity of $A$, $\f^A$ is also dominant. Thus, the Jacobian matrix
  $\jac((\f^A)^{[j]})$ has rank $p-1$ outside of a proper Zariski closed subset of $\C^n$. By the
  generalised Cramer's Rule for full rank matrices~\cite{gencramerrule}, we can find a basis of the
  null space of the Jacobian whose entries are rational functions of linear combinations of the minors
  of this matrix with a common denominator. This common denominator is itself the determinant of the
  submatrix comprised of the first $p-1$ columns of $\jac((\f^A)^{[j]})$. Hence, the numerator and
  denominator of each entry has degree at most the degree of the maximal minors which is at most
  $(p-1)(d-1)$. 
\end{proof}

We shall prove our main degree result by applying B\'ezout's Theorem~\cite[Theorem~1]{heintz} to
the finite lists of polynomials defined in Algorithm~\ref{algo:acv_minor}. Therefore, we first
give a lemma that will bound the degree of the lists $G$ and $G^\prime$ that are defined
in Algorithms~\ref{algo:acv_e} and~\ref{algo:acv_minor} respectively. We recall the
construction introduced in Corollary~\ref{cor:msj} and improved upon in Theorem~\ref{thm:algo3}.
Let $A \in \GL_n(\K)$ and let $\br \in \K$ so that the genericity assumptions of
Theorem~\ref{thm:algo3} hold. For fixed $1\leq j\leq p$ and a basis $B$ of the null space of
$\jac((\f^A)^{[j]})$, define $v_1(z), \dots, v_{n-p+1}(z) = \nabla f^A_j B$. Then,
\begin{align*}
  G &= \numer(f^A_1(\tau_1(z))-c_1,\ldots,f^A_p(\tau_1(z))-c_p, \\
  &\qquad r_1z_1e - v_1(\tau_1(z)), \dots, r_{n-p+1}z_1e - v_{n-p+1}(\tau_1(z))), \\
    G^\prime &=  \numer(f^A_1(\tau_1(z))-c_1,\dots,f^A_p(\tau_1(z))-c_p, \\
    &\qquad r_2v_1(\tau_1(z)) - r_1v_2(\tau_1(z)),\dots, 
    r_{n-p+1}v_1(\tau_1(z)) - r_1v_{n-p+1}(\tau_1(z))).
\end{align*}
  
\begin{lemma}\label{lem:degG}
  Let $\f = (f_1, \ldots, f_p) \in \K[z_1, \ldots, z_n]^p$ be a dominant polynomial mapping.
  For fixed $1 \leq j \leq p$, let $G$ and $G^\prime$ be finite lists of polynomials defined as above.
  Then, the degree of the highest dimension components of the algebraic sets $\V(G)$ and $\V(G^\prime)$
  are at most $(p(d-1)+2)^{n-p+1}(d+1)^p$ and $(p(d-1))^{n-p}(d+1)^p$ respectively.
\end{lemma}

\begin{proof}
  By \cite[Chapter~4, Section~6, Theorem~2]{CLO}, the variety $\V(G)$ can be expressed in terms of its
  irreducible components $\V(G) = W_1 \cup \cdots \cup W_k$. Since we are only concerned with the
  components of highest dimension, we may apply B\'ezout's Theorem~\cite[Theorem~1]{heintz} directly
  to bound this degree and similarly for the algebraic set $\V(G^\prime)$.

  Firstly, note that the degrees of the polynomials $f_1^A, \dots, f_p^A$ are still bounded by $d$ since
  $A$ is a generic linear change of coordinates. Then for $1 \leq i \leq p$, the degrees of the
  numerator and denominator of $f_i^A(\tau_1(z))$ are bounded by $d$. Therefore, the degree of the
  numerator of $f_i^A(\tau_1(z)) - c_i$ is at most $d+1$.

  We must now bound the degrees of $v_1, \dots, v_{n-p+1}$. To do this, first note that the degrees of
  the components of the gradient of $f_j^A$ are bounded by $d-1$. Then, we invoke
  Lemma~\ref{lem:basisdeg} to bound the degree of the entries of a basis $B$ of the null space of
  $\jac((\f^A)^{[j]})$. Thus, $v_1, \dots, v_{n-p+1}$, the dot products of the basis vectors with the
  gradient of $f_j^A$, have degree at most $p(d-1)$.

  Consider the list $G$. After applying the transformation $\tau_1(z)$ to the rational functions
  $v_1, \dots, v_{n-p+1}$, we obtain a rational function whose numerator and denominator have
  degree $p(d-1)$. Thus, for $1 \leq i \leq n-p+1$, the polynomial
  $\numer(r_i z_1 e - v_i(\tau_1(z)))$
  has degree $p(d-1)+2$. Therefore, the list $G$ has $n-p+1$ polynomials of degree $p(d-1)+2$ and
  $p$ polynomials of degree $d+1$.

  Now consider the list $G^\prime$. Recall that the entries of the basis $B$ have a common denominator.
  Therefore, the rational functions $v_1, \dots, v_{n-p+1}$ also have a common denominator.
  Thus, when we now compute the minors of the matrix,
  \[
    \begin{bmatrix}
      v_1(\tau_1(z)) & \cdots & v_{n-p+1}(\tau_1(z)) \\
      r_1 & \cdots & r_{n-p+1}
    \end{bmatrix},
  \]
  where $r_1, \dots, r_{n-p+1}$ are generic elements of the base field, we again find a common
  denominator that does not increase the degree of the numerator. Hence, for $2 \leq i \leq n-p+1$, the
  degree of $\numer(r_i v_1(\tau_1(z)) - r_1 v_i(\tau_1(z)))$ is $p(d-1)$. Then, the list $G^\prime$
  has $n-p$ polynomials of degree $p(d-1)$ and $p$ polynomials of degree $d+1$. One now applies
  B\'ezout's Theorem~\cite[Theorem~1]{heintz} to get the result.
\end{proof}

Now that we have bounded the degrees of the objects considered in our algorithms, we may now
prove our main degree result.

\begin{reptheorem}{prop:degree}
  Let $\f = (f_1, \ldots, f_p) \in \K[z_1, \ldots, z_n]^p$ be a dominant polynomial mapping.
  Let $d = \max_{1 \leq i \leq p} \deg f_i$.
  Then the asymptotic critical values of $\f$ are contained in a
  hypersurface of degree at most $p^{n-p+1}(d-1)^{n-p}(d+1)^{p}$.
 \end{reptheorem}

\begin{proof}[Proof of Theorem~\ref{prop:degree}]
  Let $A \in \GL_n(\K)$ and let $\br \in \K$ so that the genericity assumptions of Theorem~\ref{thm:algo3}
  hold. We consider the following finite lists of polynomials computed in Algorithm~\ref{algo:acv_minor}.
  \begin{align*}
    G^\prime &=  \numer(f^A_1(\tau_1(z))-c_1,\ldots,f^A_p(\tau_1(z))-c_p, \\
    &\qquad r_2v_1(\tau_1(z)) - r_1v_2(\tau_1(z)),\ldots, 
    r_{n-p+1}v_1(\tau_1(z)) - r_1v_{n-p+1}(\tau_1(z))), \\
    G_s^\prime &= \Saturate(G^\prime, z_1\numer(\delta(\tau_1(z)))), \\
    L^\prime &= G_s^\prime \cup \{z_1\}, \\
    V_j^\prime &= \Eliminate(L^\prime, \{\bz\}, \{\bc\}).  
  \end{align*}
  By Lemma~\ref{lem:ksat} and Theorem~\ref{thm:algo3}, $K_\infty(\f) \subset \bigcup_{j=1}^p \V(V_j^\prime)$.
  Thus, the degree of $K_\infty(\f)$ is bounded by the degree of $\bigcup_{j=1}^p \V(V_j^\prime)$. We now aim
  to use the algebraic description of $V_j^\prime$ as detailed above to give a bound on the degree of $\V(V_j^\prime)$.
  Then, the degree of $K_\infty(\f)$ will be bounded by $p$ times that bound.

  By \cite[Chapter~4, Section~6, Theorem~2]{CLO}, the variety $\V(G^\prime)$ can be expressed in terms of its
  irreducible components $\V(G^\prime) = W_1 \cup \cdots \cup W_k$.
  By~\cite[Chapter~4, Section~4, Theorem~10]{CLO}, the saturation
  \\\[\ideal{G_s^\prime} = \ideal{G}:\ideal{z_1\numer(\delta(\tau_1(z)))}^\infty\] corresponds to the variety
  \[\V(G_s^\prime) = \overline{\V(G^\prime) \setminus \V(z_1\numer(\delta(\tau_1(z))))}.\] Thus, $\V(G_s^\prime)$ is the
  union of a subset of the irreducible components of $\V(G^\prime)$, $\V(G_s^\prime) = W_{i_1} \cup \cdots \cup W_{i_{\ell}}$
  such that $W_{i_j} \nsubseteq \V(z_1\numer(\delta(\tau_1(z))))$ for $1 \leq j \leq \ell$. Hence, we can bound
  the degree of $\V(G_s^\prime)$ by the so-called strong degree of $\V(G^\prime)$, that is the sum of the degrees
  of the equidimensional components. 
  Similarly, by~\cite[Chapter~4, Section~3, Theorem~4]{CLO}, adding
  polynomials to an ideal equates to intersecting the respective varieties. Therefore, by B\'ezout's
  Theorem~\cite[Theorem~1]{heintz} and since the polynomial, $z_1$, that we add has degree $1$ we conclude that the
  degree of $\V(L^\prime)$ is bounded by the degree of $\V(G_s^\prime)$. Furthermore, by \cite[Lemma~2]{heintz},
  since the projection is an affine map and since $\V(L^\prime)$ is a constructible set in the Zariski
  topology, the degree of $\V(V_j^\prime)$ is bounded above by the degree of $\V(L^\prime)$.
 
  Note that we only need to consider the degree of the components of
  $\V(G^\prime)$ of highest dimension since $\V(G_s^\prime)$ is contained in the union of these components.
  In summary, we can bound the degree of $K_1^j(\f^A) \subseteq \V(V_j^\prime)$ by bounding the degree of
  $\V(G^\prime)$. Thus, by Lemma~\ref{lem:degG}, the degree of $K_1^j(\f^A)$ is at most $(p(d-1))^{n-p}(d+1)^p$.
  Therefore, the degree of $K_\infty(f)$ is at most $p^{n-p+1}(d-1)^{n-p}(d+1)^p$.
\end{proof}


%% file: proof_comp.tex
\subsection{Proof of Theorem~\ref{prop:comp}}
\label{sec:comp}
In this subsection, we assess the worst-case 
complexity of the
three algorithms given in this paper. We apply the complexity results
attained by the \geo algorithm given in~\cite{GML}.

Let $\Mult(n)$ be a cost function for multiplying two univariate
polynomials of degree at most $n$ in terms of operations in the base
field.
For instance, $\Mult(n) = O(n \log n \log\log n)$ using the
Cantor--Kaltofen algorithm~\cite{CantorK1991}.

We also denote by $\omega$, $2 \leq \omega \leq 3$, the linear
algebra complexity exponent. That is two matrices of size $n \times
n$ over a field can be multiplied in
$O(n^\omega)$ operations in the base field. At the time of writing, the
best upper bound for $\omega$ is $2.3728639$ due to Le Gall~\cite{LeGall}.
Denote by $\Omega = 1+\omega$ the related constant of the complexity 
exponent
of linear algebra over a ring.  
Finally, we denote the evaluation complexity of the polynomials of
$G$ and $z_1\numer(\delta(\tau_1(z)))$, defined in each Algorithm, by $L$.   

Given a dominant polynomial mapping $(f_1, \dots, f_p) \in \K[\bz]$ as input, for each choice of
$1 \leq j \leq p$, the first steps of both the algorithms designed in this paper construct polynomials
that will be the input of the algebraic elimination algorithms we use. In the case of our first
algorithm, these polynomials are those described in Corollary~\ref{cor:msj}. For the second algorithm,
there are small differences described in Theorem~\ref{thm:algo2}. For both algorithms, the considered
polynomials involve the computation of rational functions $v_1, \dots, v_{n-p+1}$ and polynomial
$\delta$. We recall the proof of Corollary~\ref{cor:msj} where these are defined. Let $B$ be a basis of
the kernel of the Jacobian matrix $\jac((\f^A)^{[j]})$. Then, define $v_1(z), \dots, v_{n-p+1}(z)$ to
be such that $v_i(z)$ is the product of the gradient of $f_j^A$ with the $i$th element of the basis $B$
and let $\delta$ be the determinant of the first $p-1$ columns of $\jac((\f^A)^{[j]})$. We give a lemma
that bounds the number of arithmetic operations in $\K$ required to compute $v_1, \dots, v_{n-p+1}$ and
$\delta$. 

\begin{lemma}\label{lem:evalinter}
  Let $(f_1, \dots, f_p) \in \K[\bz]$ be a dominant polynomial mapping. Let the rational functions
  $v_1, \dots, v_{n-p+1}$ and polynomial $\delta$ be defined as above. Then, computing
  $v_1, \dots, v_{n-p+1}$ and $\delta$ requires at most $O((p-1)^{n+\omega}n^{p-1}(d-1)^n)$ arithmetic
  operations in $\K$.  
\end{lemma}

\begin{proof}[Proof of Lemma~\ref{lem:evalinter}]
  We consider a Jacobian matrix of size $(p-1) \times n$ whose entries have degree $d-1$ in $n$
  variables and we want to compute a basis of its null space. It consists of $n-p+1$ vectors
  of rational functions. By~\cite{gencramerrule}, the entries of these vectors have a common
  denominator, the polynomial $\delta$, that is a maximal minor of $\jac((\f^A)^{[j]})$. Moreover, each
  numerator is a linear combination of maximal minors. Thus, these vectors are rational functions whose
  numerators and common denominator have degree $(p-1)(d-1)$ in $n$ variables. 

  We first consider the denominator, $\delta$, and proceed by Kronecker
  substitution~\cite[Chapter 8.4]{GathenG2013}.
  Since $z_1$ appears with degree at most $(p-1)(d-1)$, we set
  \[z_2 = z_1^{(p-1)(d-1)+1}.\]
  Likewise, $z_2$ appears with degree at most $(p-1)(d-1)$ so we can set
  \[
    z_3 = z_2^{(p-1)(d-1)+1} = z_1^{((p-1)(d-1)+1)^2}\] and so on until
  \[z_n = z_1^{((p-1)(d-1)+1)^{n-1}}.\]
  The result is a univariate polynomial whose
  highest degree monomial comes from
  \[z_n^{(p-1)(d-1)} = z_1^{((p-1)(d-1)+1)^{n-1}(p-1)(d-1)},\]
  so of degree $O((p-1)^n (d-1)^n)$. By the same Kronecker substitution,
  the entries of the minor corresponding to the denominator
  can be seen as univariate polynomials whose highest degree monomials comes from
  \[z_n^{d-1} = z_1^{((p-1)(d-1)+1)^{n-1}(d-1)}.\]
  Therefore, they are of degree 
  $O((p-1)^{n-1} (d-1)^n)$.

  By fast multi-point evaluation techniques, we evaluate these $(p-1)^2$ entries of the minor
  in $O((p-1)^n (d-1)^n)$ points in $O^\sim((p-1)^{n+2} (d-1)^n)$ operations in the base field
  ~\cite[Chapter 10.1]{GathenG2013}.
  
  We now perform a Gaussian elimination of size $(p-1) \times (p-1)$ for each of these evaluations in
  $O((p-1)^n (d-1)^n (p-1)^\omega) = O ((p-1)^{n+\omega} (d-1)^n)$
  operations~\cite[Chapter 12.1]{GathenG2013}. It remains to interpolate the determinant as a
  univariate polynomial in $O^\sim((p-1)^n (d-1)^n)$ operations~\cite[Chapter 10.2]{GathenG2013}.
  
  In total, the most expensive step is performing the Gaussian eliminations. Therefore,
  the overall complexity of computing $\delta$ is $O ((p-1)^{n+\omega} (d-1)^n)$. 

  We then follow the same steps to compute each of the numerators. 
  Essentially, this involves computing
  the determinants of all $\binom{n}{p-1}$ maximal minors of $\jac((\f^A)^{[j]})$. To retrieve the
  numerators, we must simply compute linear combinations of these determinants and so this step is
  negligible to the complexity. Hence, approximating $\binom{n}{p-1}$ by $n^{p-1}$ as $n \to \infty$,
  repeating this evaluation--interpolation method for each minor requires
  $O((p-1)^{n+\omega}n^{p-1}(d-1)^n)$ operations.

  Now that the basis $B$ is computed, it remains to find its product with the gradient $\d f_j$.
  Thus, by the same evaluation--interpolation techniques as before, we may compute each $v_i$ in
  $O^\sim((p-1)^n(d-1)^{n+1})$ operations. This is negligible compared to the complexity of the
  Gaussian elimination step as before and so the overall complexity of computing
  $v_1, \dots, v_{n-p+1}$ and $\delta$ is $O((p-1)^{n+\omega}n^{p-1}(d-1)^n)$.
\end{proof}

Once the initialising polynomials have been computed, the remaining steps of both algorithms designed
in this paper rely on algebraic elimination algorithms. In particular we shall use the \geo
algorithm given in~\cite{GML} in combination with the lifting algorithm of~\cite{schost2003geores}
to compute a parametric system whose solution set contains the set of asymptotic critical values of a
given dominant polynomial mapping. Since both of the algorithms we analyse will use this framework,
and in the end have essentially the same theoretical complexity, we give a full complexity analysis
of Algorithm~\ref{algo:acv_minor} and then show how the result generalises to
Algorithm~\ref{algo:acv_e}.

First, we recall the representation that is the output of \geo algorithm of~\cite{GML}.
Consider polynomials $g_1, \dots, g_m, h$ in the polynomial ring $\K[x_1, \dots, x_m]$ and the zero-dimensional algebraic set $S$ defined by $g_1 = \dots = g_m = 0$, $h \neq 0$. Let $D$ be the degree
of this set and let $T$ be a linear form of the input variables $x_1, \dots, x_m$. Then, the output
of the \geo algorithm is a representation
\[\begin{cases}
  q(T) & =  0 \\ 
  q^\prime(T) x_1 & =  v_1(T) \\
       & \; \; \vdots \\
  q^\prime(T) x_m &  =  v_m(T),
\end{cases}\]
where $q\in\Q[T]$ is a univariate polynomials of degree at most $D$
and $v_1, \dots, v_m \in\Q[T]$ are univariate polynomials of degree strictly less than $D$. 
This is a representation of the set $S$ outside of Zariski closed set $\V(q^\prime)$. We now
give our main complexity result.

\begin{reptheorem}{prop:comp}
  Let $\f = (f_1, \ldots, f_p) \in \K[\bz]^p$ be a dominant polynomial mapping 
  and let $d = \max_{1 \leq i \leq p} \deg f_i$.
  
  There exists an algorithm which, on input $\f$, computes a
  non-zero polynomial $g$
  in $\K[\bc]$ such that $ K_{\infty}(\f) \subset \V(g)$ using at most 
  \[
    O^{\sim}\left( p (p(d-1))^{2(p+1)(n-p)} (d+1)^{2p(p+1)} \right) 
  \]
  arithmetic operations in $\K$. 
\end{reptheorem}

\begin{proof}[Proof of Theorem~\ref{prop:comp}]
  We shall use Algorithm~\ref{algo:acv_minor} which by Theorem~\ref{thm:algo3} terminates, and returns
  a finite basis whose zero set has codimension at least $1$ in $\C^p$ and contains the set of
  asymptotic critical values of~$\f$. First, fix some $1 \leq j \leq p$. Then, the first
  $8$ steps of Algorithm~\ref{algo:acv_minor} are to compute the finite list of polynomials
  $G^\prime$, which we denote $h_1, \dots, h_n$, and the polynomial $\delta$. This has been
  analysed in Lemma~\ref{lem:evalinter}. 

  We note that $G^\prime$ consists of $n$ polynomials in $n+p$ variables and, by
  Theorem~\ref{thm:algo3}, defines a set of dimension $p$. Let $a_1, \dots, a_p \in \K$ be generic
  elements of $\K$. Then, substituting $c_i$ for $a_i$ in the polynomials of $G^\prime$ defines
  a zero-dimensional constructible set.
  
  Using the \geo algorithm of~\cite{GML}, we compute a representation of the system
  $h_1 = \dots = h_n = 0, \; \numer(\delta(\tau_1(z_1))) \neq 0, \; c_1 = a_1, \dots, c_p = a_p$.
  By Lemma~\ref{lem:degG}, this system has degree at most $(p(d-1))^{n-p}(d+1)^p$. For ease
  of notation, we denote this degree $D$. Then, we have the representation
    \[\begin{cases}
    q(T) & =  0 \\ 
    q^\prime(T) z_1 & =  v_1(T) \\
         & \; \; \vdots \\
    q^\prime(T) z_n &  =  v_n(T),
  \end{cases}\]
  where the polynomials $q, v_1, \dots, v_n \in \K[T]$ have degree at most
  $(p(d-1))^{n-p}(d+1)^p$. Now, using the lifting algorithm of~\cite{schost2003geores} we
  obtain a parametric representation
  \[\begin{cases}
    Q & =  0 \\ 
    \frac{\partial Q}{\partial T} z_1 & =  V_1 \\
         & \; \; \vdots \\
    \frac{\partial Q}{\partial T}z_n &  =  V_n,
  \end{cases}\]
  where $Q,V_1, \dots, V_n \in \K(\bc)[T]$ are polynomials in the linear form $T$ with
  coefficients in $\K(\bc)$. By~\cite[Theorem 1]{schost2003geores}, the numerators and
  common denominators of the polynomials $Q, V_1, \dots, V_n$ have degree at most $D$ in
  $\bc$.

  Thus, with $G^\prime_s$ defined as in Algorithm~\ref{algo:acv_minor}, this parameterises the set
  $\V(G^\prime_s)$. To now compute $V_j$, we must intersect with $\V(z_1)$ and project onto the
  $\bc$-space. First, however, we convert our representation into one described by
  polynomials in $\K[\bc, T]$. To do so, we simply multiply each polynomial
  $Q, \frac{\partial Q}{\partial T}, V_1, \dots, V_n$ by their respective common denominators.
  We denote the resulting polynomials $\tilde{Q}, \tilde{\frac{\partial Q}{\partial T}}$ and
  $\tilde{V_1}, \dots, \tilde{V_n}$ respectively. By~\cite[Theorem 1]{schost2003geores}, these
  polynomials have degree at most $D$ in $\bc$ and degree at most $D$ in $T$.
  We claim that the polynomial
  \[ g_j = \res_T\left(\tilde{Q},
      \tilde{\frac{\partial Q}{\partial T}}\right)
    \res_T\left(\tilde{Q}, \tilde{V_1}\right) \]
  defines an algebraic set that contains $\V(V_j)$. To see this, first note that the zero set
  of the polynomial $\res_T(\tilde{Q}, \tilde{V_1}) \in \K[\bc]$ contains all the points in the
  projection of the intersection with $\V(z_1)$. Thus, this zero set contains $\V(V_j)$ whenever 
  this parametrisation is defined. Then, the zero set of
  $\res_T(\tilde{Q}, \tilde{\frac{\partial Q}{\partial T}})$
  contains all the points of $\V(V_j)$ where this parametrisation is not defined. Thus, the zero
  set of the product of these polynomials contains $\V(V_j)$. One computes such a polynomial for each
  choice of $j$ and returns the product $g = \prod_{j=1}^p g_j$.

  We now analyse the complexity of the algorithm described above. Firstly, by Lemma~\ref{lem:evalinter},
  the initialisation step of the algorithm requires at most \[O((p-1)^{n+\omega}n^{p-1}(d-1)^n)\]
  arithmetic operations in $\K$ for each choice of $1 \leq j \leq p$.
  Substituting $\tau_1(z)$, finding the numerator and specialising $\bc$ is negligible,
  so we now consider the complexity of the \geo algorithm.

  By~\cite[Theorem 1]{GML}, computing a geometric resolution of the set $\V(G_s^\prime)$
  requires at most \[O(n(n L+n^\Omega)\Mult(D)^2)\] arithmetic operations in $\K$.
  Assuming that $d \geq 2$ is fixed, we may bound the evaluation complexity, $L$, by
  $n \binom{n + D}{n} = O(n^{d+1})$. Thus, by excluding logarithmic factors, we arrive at a
  simplification to the class
  \[O^\sim(n^{d+3}D^2).\] 

  Furthermore, by~\cite[Theorem 2]{schost2003geores}, the lifting step requires at most
  \[O^\sim((n L+n^4)\Mult(D)\Mult(D^p) + n p^2 D \Mult(D)\Mult(D^{p-1})) = O^\sim((n
    L+n^4+n p^2)D^{p+1})\]
  arithmetic operations in $\K$. By again assuming that $d \geq 2$ is fixed, we apply the same
  simplification of the evaluation complexity. Furthermore, assuming that $n \geq p$, we arrive
  at a simpler form:
  \[O^\sim(n^{d+2}D^{p+1}).\] 
  
  The final step of importance is to compute $2p$ resultants. Recall that $\tilde{Q}, \tilde{V_1}$ and
  $\tilde{\frac{\partial Q}{\partial T}}$ have degree at most $D$ in $\bc$ and $T$. Thus, each
  Sylvester matrix has at most $2D$ columns and has entries of degree at most $D$. Hence, the
  determinant of these matrices, the resultants we wish to compute, have degree at most $2D^2$.
  We return to a Kronecker substitution to reduce to the bivariate case, leaving the variables
  $c_1$ and $T$. Since the variables $\bc$ each occur with degree at most $2D^2$, we can set
  \[c_2 = c_1^{2D^2+1}, \dots, c_p = c_1^{(2D^2+1)^{p-1}}.\]
  Therefore, with this substitution, we can write the entries of each Sylvester matrix as
  univariate polynomials in $c_1$ with degree in the class $O(2^p D^{2p})$.

  By~\cite[Corollary 11.21]{GathenG2013}, we can compute each bivariate resultant within
  $O^\sim(2^p D^{2p+2})$ arithmetic operations in $\K$. We compute $2p$ resultants and so the overall
  complexity of computing the resultants is in the class \[O^\sim(p2^{p+1}D^{2p+2}).\]
  
  In summary, the overall complexity
  for computing a polynomial whose zero set contains the asymptotic critical values of $\f$ is in the
  class:
  \[O^\sim(p(p-1)^{n+\omega}n^{p-1}(d-1)^n + n^{d+3}D^2 + n^{d+2}D^{p+1} + p2^{p+1}D^{2p+2}).\]
  The complexity of the resultant computation is dominant. Hence, this simplifies to the class:
  \[O^\sim \left( p (p(d-1))^{2(p+1)(n-p)} (d+1)^{2p(p+1)} \right).\qedhere\]
\end{proof}

\begin{corollary}\label{cor:algo2}
  Let $\f = (f_1, \ldots, f_p) \in \K[\bz]^p$ be a dominant polynomial mapping 
  and let $d = \max_{1 \leq i \leq p} \deg f_i$. Then, Algorithm~\ref{algo:acv_e} returns
  a polynomial $g$ in $\K[\bc]$ such that $ K_{\infty}(\f) \subset \V(g)$ within
  \[
    O^\sim\left(  p (p(d-1))^{2(p+1)(n-p)} (d+1)^{2p(p+1)}  \right)
  \]
  arithmetic operations in $\K$.
\end{corollary}

\begin{proof}[Proof of Corollary~\ref{cor:algo2}]
  Note that the procedure described in proving Theorem~\ref{prop:comp} can be used for
  Algorithm~\ref{algo:acv_e} with very few adjustments required. Firstly, note that the initialisation
  step is almost identical except that we consider $n+1$ polynomials in $n+2$ variables rather than
  $n$ polynomials in $n+1$ variables. Therefore, the additional variable $e$ found in
  Algorithm~\ref{algo:acv_e} will not change the complexity class. Furthermore, by Lemma~\ref{lem:degG},
  the list of polynomials considered in Algorithm~\ref{algo:acv_e} defines an algebraic set of degree
  at most $(p(d-1)+2)^{n-p+1}(d+1)^p$. Comparing this to $D$, the degree of the algebraic set defined by
  $G^\prime$ in Algorithm~\ref{algo:acv_minor}, which is equal to $(p(d-1))^{n-p}(d+1)^p$, we
  conclude that this difference also will not change the complexity class. The last main difference
  is the resultant step. Indeed, we must also consider the intersection with the variety $\V(e)$.
  This results in the computation of one more resultant per choice of $j$, $\res_T(\tilde{Q},
  \tilde{V_e})$ where $V_e$ is the lifted parametrisation of $e$. However, the difference between
  computing $2p$ resultants and $3p$ resultants does not change the complexity class. 
\end{proof}


%% file: apps.tex
\section{Applications}\label{sec:apps}
\subsection{Solving Polynomial Optimisation Problems}\label{sec:pop}
In this subsection we present how to use the algorithms detailed in this paper to solve global polynomial optimisation problems.

Firstly, we review to problem we wish to solve. Consider a polynomial $f \in \Q[\bz]$. We aim to compute the global infimum of this polynomial $\inf_{\bx \in \R^n} f(\bx) = f^* \in \R \cup \{-\infty\}$. We can solve this problem exactly by computing the generalised critical values of $f$.

There are three cases:
\begin{itemize}
\item $f^*$ is reached. Then, $f^*$ is a critical value of $f$;
\item $f^*$ is reached only at infinity, meaning that there is no minimiser $\bx \in \R^n$ but instead a path $\bx_t \in \R^n$ that approaches the infimum as $\norm{\bx_t}\to\infty$. Then, $f^*$ is an asymptotic critical value of $f$;
\item $f^* = -\infty$.
\end{itemize}

The procedure is as follows: We first compute an algebraic representation of the 
generalised critical values of $f$. We do this by computing a polynomial whose roots contain
the asymptotic critical values by using the algorithms described in 
this paper or in the papers \cite{jelonek2014reaching,K2,polarcurve}. Then, using the gradient ideal as 
in~\cite{faugere2012critical} we can similarly compute a polynomial whose roots contain the
critical values of $f$. There are algebraic elimination algorithms that compute such polynomials with
rational coefficients, for example \gbs~\cite[Chapter~2]{CLO} or the \geo algorithm designed
in~\cite{GML}, since we assumed that $\f \in\Q[\bz]$. Thus, after finding a common denominator,
we may assume these polynomials have integer coefficients. Then, we may use a
real root isolation algorithm such as in~\cite{efficientrealroots}, based on Descartes' rule of sign
~\cite[Theorem 2.44]{BPR}, to compute isolating intervals with rational endpoints for all real roots
of these polynomials.

Let $C = \{c_1, \dots, c_k\} \subset \R$ be the finite set of real algebraic numbers that are the
real roots of the above polynomials. Then, the set $C$ contains the generalised critical values of $f$.
By~\cite[Theorem 3.1]{KOS}, the polynomial $f$ with restricted domain $f: \R^n \setminus f^{-1}(K(f))
\to \R \setminus K(f)$ is a fibration over each connected component of $\R \setminus K(f)$. Therefore,
since $C$ is finite, the same fibration property applies to the restriction
$f: \R^n \setminus f^{-1}(C) \to \R \setminus C$. Hence, to decide the emptiness of each connected
component of $\R \setminus C$, it is sufficient to decide the emptiness of one fibre for each connected
component.

After computing the isolating intervals for the elements of $C$, we may now choose rational numbers
$r_1, \dots, r_k$ so that \[ r_1 < c_1 < r_2 < \cdots < r_k < c_k.\]

We must assess the emptiness of the fibres of these values. We do so using the algorithm
designed in~\cite{safeyschost2003}. We consider, for $0 \leq i \leq k$, the ideal $\ideal{f-r_i}$.
This algorithm requires a radical ideal such that $\V(f-r_i)$ is smooth and equidimensional.
Since $r_i$ is outside of these isolating intervals, we have that $\V(f-r_i)$ is smooth and
equidimensional. Furthermore, since $\V(\sqrt{\ideal{f-r_i}}) = \V(f-r_i)$, we may instead consider the
square-free part of $\ideal{f-r_i}$, $\sqrt{\ideal{f-r_i}}$, to decide the emptiness of $\V_\R(f-r_i) = \V(f-r_i) \cap \R^n$. 

Firstly, if $\V_\R(f-r_0)$ is non-empty then we must be in the third case and so $f^* = -\infty$.
For the remaining two cases, let $i$ be the least index such that $\V_\R(f-r_i)$ is non-empty, if
such an index exists. If $r_i$ is greater than the least critical value, which one may decide from the
isolating intervals, then the least critical value is the minimum of $f$. Else, $c_{i-1}$ corresponds
to an asymptotic critical value and is the infimum of $f$. If such an index does not exist, then the
least critical value of $f$ is the minimum and if $f$ does not have any critical values, then the
infimum is $c_k$.

The complexity of the algorithm for polynomial optimisation described is as follows. For a
polynomial $f \in \Q[\bz]$ of degree $d$, we first compute
the polynomial representation of $K(f)$. By Theorem~\ref{prop:degree} and~\cite[Theorem 4.3]{msed2007testing}, we can compute this within $O^\sim(n^7d^{4n})$ arithmetic operations in $\Q$.
By~\cite[Corollary 4.4]{K2}, $f$ has at most $d^n$ generalised critical values. Thus, with $\beta$
bounding the bit-size of the input polynomial, isolating the real roots with the algorithm designed
in~\cite{efficientrealroots} requires $O(\beta d^{4n})$ operations. We must then choose at most
$d^n + 1$ points in $\Q$, the $r_1, \dots, r_{d^n}$ as above, and decide the emptiness of
each $\V_\R(f-r_i)$. This requires the use of the algorithm designed in~\cite{safeyschost2003} at most
$d^n$ times with each computation requiring $O(n^7d^{3n})$ operations. Thus, one can compute
an isolating interval for the infimum of a polynomial $f \in \Q[\bz]$ of degree $d$ in
$O^\sim(n^7d^{4n})$ arithmetic operations in $\Q$.

\begin{example}
  Consider the polynomial $f = z_1^2z_2^2 + 2z_1z_2^3 + z_2^4 + z_1^2
  + 3z_1z_2 + 2z_2^2$. First, we compute the set of generalised critical values. Note that in this
  simple example it is possible to find exactly the real algebraic numbers that contain the generalised
  critical values because the degrees of the polynomials we compute in our algorithms are small.
  We find that $K_0(f) =\{0\}$ and using Algorithm~\ref{algo:acv_e} we find
  $K_\infty(f) \subset \{-\frac{1}{4}\}$. Now, to show that $f^* = -\frac{1}{4}$ one must first show
  that $f$ is bounded from below. To do so, decide the emptiness of the real
  variety $\V_\R(f-r)$ for some real number $r < -\frac{1}{4}$. For example, we can choose $r = -1$ and
  find that this variety is indeed empty. Finally, one must show that $-\frac{1}{4}$ truly is an 
  asymptotic critical value as Algorithm~\ref{algo:acv_e} computes a superset of the asymptotic
  critical values. Thus, one shows that $f$ takes values less than $0$ by once again deciding the
  emptiness of a fibre. So, consider the variety $\V_\R(f + \frac{1}{8})$ and find that it is not 
  empty. This shows that $f$ takes values less than $0$ and by the fibration property satisfied by the 
  generalised critical values we conclude that the infimum of $f$ is $-\frac{1}{4}$.
\end{example}

\begin{example}
  Consider the polynomial $f = z_1^3 + z_1^2z_2^2 -2z_1z_2 + 1$. We
  find that $K_0(f) =\{1\}$ and $K_\infty(f) \subset \{0\}$. We first test the third case.
  Take a value less than $0$, for example $-1$, and decide the emptiness of $\V_\R(f+1)$. We find
  that this fibre is not empty and so by the fibration property, we conclude that $f^* = -\infty$.
\end{example}

For more information on solving polynomial optimisation problems, we
refer to~\cite{msed2008computing, greuetsafey, schweighofer2006global}.

\subsection{Deciding the emptiness of semi-algebraic sets defined by a single inequality}\label{sec:emptiness}
In this subsection, we continue to explore the applications of algorithms
computing generalised critical values. Let $f \in \Q[\bz]$
be a polynomial with degree $d$ and consider the semi-algebraic set
$S$ defined by the single inequality $f > 0$. The goal is to test the
emptiness of the set $S$ and in the case that $S$ is not empty to
compute at least one point in each connected component. There exists
$e \in \Q^+$ small enough such that the problem is reduced to computing
at least one point in each connected component of the real algebraic set
$\V_\R(f-e)$. Such an $e$ is small enough in this sense if it is
less than the least positive generalised critical value of the map
$z \in \R^n \to f(z) \in \R$, we refer to~\cite[Theorem 5.1]{msed2007testing}.
To decide when this is the case, one computes isolating intervals for the
generalised critical values by~\cite[Algorithm~10.63]{BPR}. Once an appropriate $e$ has
been chosen, it remains to compute at least one point in each connected
component of  $\V_\R(f-e)$. This may be
accomplished using the algorithm designed in~\cite{safeyschost2003}. To apply
this algorithm, we require that $\ideal{f-e}$ is radical and $\V(f-e)$ is equidimensional
and smooth. Since $e$ is away from any generalised critical values we have that
$\V(f-e)$ is equidimensional and smooth. Moreover, if $\ideal{f-e}$ is not
radical, we may simply take the square-free part instead as $\V(\sqrt{\ideal{f-e}}) = \V(f-e)$. 

As in the previous application, the complexity of computing isolating intervals for all real
generalised critical values is in the class $O^\sim(n^7 d^{4n})$. After choosing an
appropriate rational number $e$, it remains to apply the algorithm designed in~\cite{safeyschost2003}.
This requires $O(n^7d^3n)$ operations. Therefore, the overall complexity of deciding the emptiness
of the semi-algebraic set defined by $f>0$ is in the class  $O^\sim(n^7 d^{4n})$. Moreover,
in the case where this set is not empty, at least one point in each connect component is computed.

\begin{example}
  Consider the polynomial $f = z_1^2(1-z_2) - (z_1z_2^2-1)^2$. Again, in this simple example
  we obtain polynomials of degree at most $2$ from our algorithms and so we can give explicitly
  the set containing the generalised critical values. The polynomial giving the asymptotic critical
  values is $c$ while for the critical values it is $229c^2-202c-27$. Hence, we find that
  $K(f) \subset \{0,1,\frac{-27}{229}\}$. We note that the value $1$ is a critical value, hence we
  may decide immediately that the semialgebraic set defined by $f > 0$ is nonempty. Now, to
  compute at least one sample point in each connected component of this set, we must choose a
  suitable fibre to investigate. Thus, we choose a rational value greater than $0$ and less than
  the least generalised critical value, such as $\frac{1}{2}$, and use the algorithm
  in~\cite{safeyschost2003} to compute sample points for each connected component of $\V_\R(f - \frac{1}{2})$.
  We may do so because $\ideal{f - \frac{1}{2}}$ is a radical ideal.
\end{example}


%% file: experiments.tex
\section{Experiments}
\label{sec:expo}

The algorithms discussed in this paper have initially been implemented
in the \Maple computer algebra system using a combination of \FGb~\cite{FGb}and \MSolve~\cite{MSolve},
both implemented in C, to perform the \gb computations as well as to compute the
degree of various objects described in the tables below. In this section, we present the
experimental results of these implementations with computations
performed on a computing server with 1536 GB of memory and an Intel Xeon E7-4820 v4 2GHz
processor. Exceptionally, we also present the timings of these implementations
when given polynomials from practice as input with computations performed on a computing
server with 754 GB of memory and an Intel Xeon Gold 6244 3.6GHz processor. The computations were performed under finite 
fields before reconstructing the rational polynomial whose roots are the asymptotic
critical values. For our timings, the entry $\infty$ has been given in the cases when the algorithm
has not terminated within 2 days. Additionally, if a computation could not be performed we give the entry N/A.
We give the remainder of the entries correct up to two significant figures.

We compare our algorithms to the one derived from the work of~\cite[Section 4]{KOS} combined with
our first element of randomisation from Lemma~\ref{lem:ksat}. We denote the resulting algorithm
$\alg$.

Our algorithms outperform $\alg$ in every tested circumstance as expected. Additionally, in general,
Algorithm~\ref{algo:acv_minor} is faster than Algorithm~\ref{algo:acv_e}. However,
the polynomial system returned by Algorithm~\ref{algo:acv_minor} can be of higher degree than that
returned by Algorithm~\ref{algo:acv_e}. Concerning the particular case we study in this section,
the case $p=1$, we have that the polynomial returned by Algorithm~\ref{algo:acv_e} is a factor of the
output of Algorithm~\ref{algo:acv_minor}. 

We also test how Algorithm~\ref{algo:acv_e} behaves when we perform the saturation step with
two different methods. While both methods use \gbs, the first method given
in~\cite[Theorem 4.4.14]{CLO} introduces a new variable which acts as the inverse of the polynomial
one wishes to saturate by. The second method, given by Bayer~\cite{Bayer} and described
in~\cite[Exercise~15.41]{Eisenbud}, works for homogeneous ideals. Thus, we also introduce a new
variable to first homogenise our ideal. Then, to perform the saturation by the ideal $\ideal{z_1}$,
we factor out all powers of $z_1$ from the homogeneous ideal. Setting the introduced variable to $1$
returns a basis for the saturated ideal. We find that for generic dense polynomials, the first method
is the best. However, for the particular families of polynomials with asymptotic critical values that
we test out algorithms with the Bayer method is noticeably faster. 

We give three families of polynomials that have asymptotic critical values for the purpose of testing
our algorithms. For $n \geq 2$, let
\[
  f_n = z_1^2 + \sum_{i=2}^n (z_1 z_i - 1)^2, \; \;
  g_n = \sum_{i=1}^n \frac{\prod_{j=1}^n z_j^2}{z_i^2}, \; \; \;\\
  m_n = \sum_{i=1}^n \prod_{j=1}^i z_j^{2^{i-j}}.
\]
For $n \geq 2$, each of these polynomials has an asymptotic critical value at $0$. For $n \geq 3$, $f_n$
also has an asymptotic critical value at $n$. Additionally, we compare our algorithms with random dense
polynomials, which do not have asymptotic critical values. To do so, we introduce the following notation.
For a random dense polynomial in $k$ variables and degree $s$ we write
$d_s n_k$.

\begin{table}
  \centering
  \begin{tabular}{|r|r|r|r|r|}
    \hline
     & $\alg$ & Algo.~\ref{algo:acv_e} & Algo.~\ref{algo:acv_e} w/ Bayer & Algo.~\ref{algo:acv_minor}\\
    \hline
    Polynomial  & time (s) & time (s) & time (s) & time (s) \\
    \hline
    $f_5$ & 1\,200 & 0.061 & 0.056 & 0.034 \\
    \hline
    $f_{25}$ & $\infty$ & 17 & 14 & 12 \\
    \hline
    $g_5$ & $\infty$ & 5.3 & 2.1 & 0.61 \\
    \hline
    $g_6$ & $\infty$ & 480 & 63 & 20.7 \\
    \hline
    $m_4$ & $\infty$ & 2.5 & 2.0 & 0.48 \\
    \hline
    $m_5$ & $\infty$ & 6\,700 & 130 & 130 \\
    \hline
    $d_2n_{20}$ & 20 & 0.29 & 0.24 & 0.23 \\ 
    \hline
    $d_2n_{100}$ & $\infty$ & $\infty$ & 120 & 100 \\ 
    \hline
    $d_3n_5$ & $\infty$ & 5.30 & 19 & 0.031 \\ 
    \hline
    $d_3n_7$ & $\infty$ & $\infty$ & $\infty$ & 0.25 \\ 
    \hline
    $d_4n_4$ & $\infty$ & 390 & 1\,300 & 0.074 \\
    \hline
    $d_4n_6$ & $\infty$ & $\infty$ & $\infty$ & 3.3 \\
    \hline
  \end{tabular}
  \caption{Timings.}
  \label{tab:timing}
\end{table}

From Table~\ref{tab:timing}, we see that Algorithm~\ref{algo:acv_e} and Algorithm~\ref{algo:acv_minor}
surpass $\alg$ in all instances by a large factor. Moreover, for our specific polynomial families,
the Bayer method of saturation is faster. However, for generic dense polynomials, the Bayer method
is slower in all cases except $d=2$. Additionally, for all tested examples, Algorithm~\ref{algo:acv_minor}
is the quickest. In particular, for generic dense polynomials with degree $4$ in $6$ variables,
Algorithm~\ref{algo:acv_minor} finishes within a few seconds while neither of the other two
algorithms terminated within 48 hours.

Next, in Table~\ref{tab:practice}, we present the timings of Algorithm~\ref{algo:acv_minor} with
polynomials coming from practice. For these computations, we now use \MSolve~\cite{MSolve}
to perform the saturation step and when applicable we complete the computation with \FGb~\cite{FGb}.
The polynomials $f_{1,2}, f_{1,3}, f_{2,2}$ and $f_{2,3}$ are given in~\cite{kaltofen2009proof} while the
polynomials $s_1, s_2$ and $s_3$ can be found on the webpage
\url{https://www-polsys.lip6.fr/~ferguson/sauter.html}.

\begin{table}
  \centering
  \begin{tabular}{|r|r|r|}
    \hline
    Polynomial & MSolve saturation (s) & FGb elimination (s)\\
    \hline
    $f_{1,2}$ & $\infty$ & N/A \\
    \hline
    $f_{1,3}$ & 50\,000 & $\infty$ \\
    \hline
    $f_{2,2}$ & 2\,800 & 150\,000 \\
    \hline
    $f_{2,3}$ & $\infty$ & N/A \\
    \hline
    $s_1$ & 4\,400 & 100 \\
    \hline
    $s_2$ & 4\,000 & 88 \\
    \hline
    $s_3$ & 3\,800 & 100 \\
    \hline
  \end{tabular}
  \caption{Timings for Algorithm 3 on examples from practice}
  \label{tab:practice}
\end{table}

Finally, in Table~\ref{tab:set_degree}, we give our results on the
degree of hypersurfaces containing the asymptotic critical values. We
compare the bounds given for Algorithm~\ref{algo:acv_minor}
in Theorem~\ref{prop:degree} to the degree of the output of Algorithm~\ref{algo:acv_minor}
and to the degree of the ideal generated by the basis $G^\prime$ as computed in
Algorithm~\ref{algo:acv_minor}. We compute the latter degree by
using the \verb+fgb_hilbert+ function of the \FGb
library~\cite{FGb} to find the numerator of the Hilbert series of
$\K[\bz] / \ideal{G^\prime}$ and evaluating this at $1$. We note that for the case of
random dense polynomials, the degree of the ideal generated by the
basis $G^\prime$ equals the bound we give on the the degree of the asymptotic
critical values.

\begin{table}[htbp!]
  \centering
  \begin{tabular}{|r|r|r|r|r|}
    \hline
    Polynomial & Algo.~\ref{algo:acv_minor} & $G$ & $K_\infty(f)$\\
    \hline
    $f_5$ & 405 & 4 & 3\\
    \hline
    $f_{25}$ & 1\,412\,147\,682\,405 & 4 & 3\\
    \hline 
    $g_5$ & 21\,609 & 90 & 1 \\
    \hline
    $g_6$ & 93\,934\,323 & 138 & 1\\
    \hline
    $m_4$ & 43\,904 & 124 & 1\\
    \hline
    $m_5$ & 25\,920\,000 & 572 & 1\\
    \hline
    $d_2n_{20}$ & 3 & 3 & 0\\
    \hline
    $d_2n_{100}$ & 3 & 3 & 0 \\
    \hline
    $d_3n_5$ & 64 & 64 & 0 \\
    \hline
    $d_3n_7$ & 256 & 256 & 0 \\
    \hline
    $d_4n_4$ & 135 & 135 & 0 \\
    \hline
    $d_4n_6$ & 1\,215 & 1\,215 & 0 \\
    \hline
  \end{tabular}
  \caption{Comparison of degree bounds and degree reached during the algorithm.}
  \label{tab:set_degree}
\end{table}